\documentclass[11pt,thmsb]{article}
\usepackage{amsfonts}
\usepackage{amssymb}
\usepackage{amsmath}
\usepackage{harvard}
\usepackage{setspace}
\usepackage{dirtytalk}
\usepackage{graphicx}
\usepackage{multirow}
\usepackage{rotating}
\usepackage{lscape}
\usepackage{caption}
\usepackage{subcaption}
\usepackage[affil-it]{authblk}
\usepackage[flushleft]{threeparttable}
\usepackage{mathrsfs}
\usepackage{float}

\renewcommand{\arraystretch}{.5}
\advance \topmargin by -\headheight
\advance \topmargin by -\headsep
\evensidemargin \oddsidemargin
\let\tilde=\widetilde
\newtheorem{theorem}{Theorem}
\newtheorem{assumption}{Assumption}
\newtheorem{lemma}{Lemma}

\renewcommand{\arraystretch}{1.5}

\input{tcilatex}
\makeatletter
\def\@biblabel#1{\hspace*{-\labelsep}}
\makeatother

\title{Hermite Polynomial-based Valuation of American Options with General Jump-Diffusion Processes\thanks{We are grateful for extensive discussions with Jerome Detemple, Iván Fernández-Val, Jean-Jacques
Forneron, Hiroaki Kaido, Pierre Perron, Zhongjun Qu, and Hao Xing. We would also like to thank Undral Byambadalai, Shuowen Chen, Taosong Deng, Anlong Qin and seminar participants at Boston University for their comments. \textit{Matlab} code to implement the numerical examples in this paper 
can be found at https://sites.google.com/view/guang-zhang/research
}}

\author{Li Chen\thanks{%
Email: lichencharlie@gmail.com. }} 
\affil{Questrom School of Business, Boston University, Boston, MA, 02215}
\author{Guang Zhang\thanks{%
Email: gzhang46@bu.edu. }} 
\affil{Department of Economics, Boston University, Boston, MA, 02215}

\date{
\today
}

\begin{document}

\maketitle

\begin{abstract}
\baselineskip=17pt 

We present a new approximation scheme for the price and exercise policy of American options. The scheme is based on Hermite polynomial expansions of the transition density of the underlying asset dynamics and the early exercise premium representation of the American option price. The advantages of the proposed approach are threefold. First, our approach does not require the transition density and characteristic functions of the underlying asset dynamics to be attainable in closed form. Second, our approach is fast and accurate, while the prices and exercise policy can be jointly produced. Third, our approach has a wide range of applications. We show that the proposed approximations of the price and optimal exercise boundary converge to the true ones. We also provide a numerical method based on a step function to implement our proposed approach. Applications to nonlinear mean-reverting models, double mean-reverting models, Merton's and Kou’s jump-diffusion models are presented and discussed.


\end{abstract}
\thispagestyle{empty}\setcounter{page}{0}\baselineskip=18pt\newpage

\section{Introduction}
The valuation of American-style options poses a challenge for both academic and industrial professionals. One of the difficulties comes from the fact that such a valuation process relies on the identification of an optimal exercise policy. So far, considerable effort has been put into simple settings where the underlying asset price follows a log-normal process and the interest rate is constant (i.e., the standard model, or the Black-Scholes model). Within this context, Kim (1990) decomposed the American option price into two parts: the corresponding European option price and an Early Exercise Premium (EEP) that captures the gains from exercising the option prior to its maturity. Similar results are provided by Jacka (1991) and Carr \emph{et al.} (1992). The EEP representation of the American option price has proved extremely useful because it provides a recursive integral equation for the optimal exercise boundary. Solving the integral equations is key to the valuation process: it identifies the optimal exercise policy, providing a parametric formula for the option price. Such an approach, based on the integral equation, is straightforward to implement and shows significant advantages over other numerical procedures such as methods based on binomial lattices, Monte Carlo simulation, and Partial Differential Equations (PDE). See Brodie and Detemple (2004) for a survey of methods on the valuation of American options.

While the valuation of American options in the standard model has been resolved, empirical evidence suggests that the log-normality assumption does not hold in reality. For example, the “volatility smile” phenomenon is a well-known pattern in option pricing practice. To allow for the consistency of models with empirical regularities, non-constant, or even non-deterministic model parameters should be considered. Unfortunately, analytical results in the standard model can not be generalized to models with stochastic parameters in a straightforward manner. Efforts have been made to solve diffusion models with nonconstant parameters. For example, Jacka and Lynn (1992) considered general contingent claims written on diffusion processes. 
Detemple and Tian (2002) presented an integral equation approach for the valuation of American-style derivatives when the underlying asset price follows a general diffusion process and the interest rate is stochastic. 
See Rutkowski (1994), and Gukhal (2001) for the valuation of American options for other non-standard models.

All the above-mentioned methods rely on the fact that the transition density of the underlying asset dynamics admits a closed functional form. Such conditions have limited the scope of stochastic processes that can be considered. Furthermore, even when the transition density exists in closed form, the structure may be quite complex, and in turn, the method may be difficult to implement. To overcome these difficulties, this article presents a systematic treatment of the valuation of American options based on Hermite polynomial expansions and the EEP formula. Therefore, our contributions to the literature are threefold.

The first contribution is that our method does not rely on the existence of analytical solutions to the transition density or the characteristic function of the distribution of the underlying asset price. Moreover, there are no requirements for affine structures. We propose using Hermite polynomials to approximate the transition density for a given jump-diffusion model. The Hermite polynomial approximation is based on Ait-Sahalia (2002, 2008), and Yu (2007). This approach gives an explicit sequence of closed-form solutions to the transition density and is shown to converge to the turn density. See Ait-Sahalia (1999), Egorov, Li, and Xu (2003), Ait-Sahalia and Kimmel (2007, 2010), and Xiu (2014) for studies related to this approach.

The second contribution is that, our method is fast and accurate, while the price and exercise policy can be jointly approximated by our approximation scheme. Owing to the inherent nature of the EEP approach, by solving the integral equations with the Hermite polynomial-based approximation to the transition density, we can generate an approximation of the optimal exercise boundary. In turn, we provide a theorem (Theorem \ref{thm:aoh} in Subsection \ref{subsection:hp}) on the convergence of our proposed approximation. When we increase the order of the Hermite polynomial in the approximation of the transition density, the proposed approximations of the price and exercise boundary of American options further improve. We can control the smoothness and accuracy of the exercise boundary by changing the polynomial order in the expansion of the transition density.

Third, our method can be easily extended to jump-diffusion models and multidimensional cases. The extension is straightforward to implement without additional theoretical/modeling complications. Kou (2002) established the analytical solutions for European option pricing in a jump-diffusion model. However, the American option pricing with jump-diffusion processes remains challenging. Gukhal (2001) derived an EEP formula for the value of American options in a jump-diffusion model. Despite all these efforts, one major drawback is that jump-diffusion models usually come without closed-form transition densities. Even if such a density exists, its functional form may be quite complicated in structure and the implementation requires a significant amount of human and computer power. Our method, on the other hand, can overcome these difficulties. By using a Hermite polynomial expansion, we can control the computational cost of the pricing algorithm by specifying the order of the expansion. Moreover, unlike conventional approaches such as the PDE-based method (finite difference, for example), Hermite polynomial expansions can be applied to a vector of stochastic processes, and the results can be directly applied to multi-dimensional models.

The structure of this paper is as follows. Section \ref{section:ao_2} describes an approach to the American option valuation when the underlying asset prices follow a general diffusion process. Section \ref{section:ao_3} describes the generalization of the method to jump-diffusion processes. Section \ref{section:ao_4} presents a numerical algorithm for implementing the proposed approach. Section \ref{section:ao_5} provides several examples to demonstrate the efficiency of the proposed approach. Finally, Section \ref{section:ao_6} concludes this paper.

\section{Valuation of American Options in Diffusion Models}\label{section:ao_2}

\subsection{American Options}

We consider the stock price $S$ defined on a probability space $\left(\varOmega,\mathcal{F},\mathbb{P}^{*}\right)$
with filtration $\mathbb{F}=\left(\mathcal{F}_{t}\right)_{0\leq t\leq T}$
satisfying the usual conditions and following:
\begin{equation}
dS_{t}=\left(r\left(S_{t};\theta\right)-\delta\left(S_{t};\theta\right)\right)dt+\sigma\left(S_{t};\theta\right)dW_{t}\label{eq:aomodel}.
\end{equation}
We also denote $\mu\left(S_{t};\theta\right)=r\left(S_{t};\theta\right)-\delta\left(S_{t};\theta\right)$.
Let $D_{S}=\left(\underline{s},\bar{s}\right)$ be the domain of the
diffusion $S$.

The arbitrage-free price of an American put option with a finite expiration $T>0$ and a strike price $K$ can be expressed as the expected
value of its discounted payoff:
\begin{equation}
P\left(t,S_{t}\right)=\sup_{t\leq\tau\leq T}\mathbb{E}^{*}\left[e^{-\left(\tau-t\right)r}\left(K-S_{\tau}\right)^{+}|S_{t}\right]\label{eq:aovalue}
\end{equation}
under the risk-neutral probability measure $\mathbb{P}^{*}$. Here
$\tau$ is the stopping time.

\subsection{Early Exercise Boundary}

Let $\mathcal{B}=\left\{ B_{t}:B_{t}\geq0,t\in\left[0,T\right]\right\} $
denote the optimal early exercise boundary of the American put option.
Then the arbitrage-free price of the American put option, $P\left(t,S_{t}\right)$, solves
the following free boundary problem:
\[
\mathcal{L}P=0,
\]
\[
P\left(T,S_{T}\right)=\left(K-S_{T}\right)^{+},
\]
\[
\lim_{S_{t}\uparrow\infty}P\left(t,S_{t}\right)=0,
\]
\[
\lim_{S_{t}\downarrow B_{t}}P\left(t,S_{t}\right)=K-B_{t},
\]
\[
\lim_{S_{t}\downarrow B_{t}}\frac{\partial P\left(t,S_{t}\right)}{\partial S_{t}}=-1,
\]
where $\mathcal{L}f=\frac{1}{2}\sigma^{2}S_{t}^{2}\frac{\partial^{2}f}{\partial S_{t}^{2}}+\left(r-\delta\right)S_{t}\frac{\partial f}{\partial S_{t}}-rf+\frac{\partial f}{\partial t}.$

\begin{theorem}[Exercise Premium Representation]\label{thm:eep}
We assume that $r$, $\delta$,
and $\sigma$ are continuously differentiable, and (\ref{eq:aomodel})
has a unique strong solution. Then, in the continuation region $\mathcal{C}$\footnote{The continuation region is the set of pairs $(S,t)$ at which immediate exercise is sub-optimal.},
the value of the American put option, $P_{0}\equiv P\left(0,S_{0}=s_{0}\right)$,
has the following early exercise premium representation:
\begin{equation}
P_{0}=p_{0}+e_{0}\label{eq:ao}
\end{equation}
where $p_{0}$ represents the price of a European put option, that is,
\begin{equation}
p_{0}\equiv p\left(0,S_{0}=s_{0}\right)=\int_{0}^{K}e^{-rT}\left(K-S_{T}\right)\psi\left(S_{T};S_{0}=s_{0}\right)dS_{T}\label{eq:ao_eo},
\end{equation}
and $e_{0}$ is the early exercise premium given by
\begin{equation}
e_{0}\equiv e\left(0,S_{0}=s_{0},B\left(\cdot\right)\right)=\int_{0}^{T}\int_{0}^{B_{t}}\left(rK-\delta S_{t}\right)e^{-rt}\psi\left(S_{t};S_{0}=s_{0}\right)dS_{t}dt\label{eq:ao_eep},
\end{equation}
and $\psi\left(S_{t};S_{0}=s_{0}\right)$ denotes the risk-neutral
transitional density function of $S_{t}$, given $S_{0}=s_{0}$. The
exercise boundary $B_{t}$ solves the recursive nonlinear integral
equation
\begin{equation}
K-B_{t}=p\left(t,B_{t}\right)+e\left(t,B_{t},B\left(\cdot\right)\right)\quad\forall t\in\left[0,T\right),\label{eq:ao_bdy}
\end{equation}
subject to the boundary condition
\[
B_{T-}\equiv\lim_{t\uparrow T}B_{t}=\min\left\{ K,\frac{r\left(B_{T};\theta\right)}{\delta\left(B_{T};\theta\right)}K\right\}.
\]
At maturity, $B_{T}=K\geq B_{T-}$. The functions $p$ and $e$ in
(\ref{eq:ao_bdy}) are defined as following:
\begin{equation}
p\left(t,B_{t}\right)\equiv\int_{0}^{K}\left(K-S_{T}\right)\psi\left(S_{T};S_{t}=B_{t}\right)dS_{T},\label{eq:ao_eo2}
\end{equation}
\begin{equation}
e\left(t,B_{t},B\left(\cdot\right)\right)\equiv\int_{t}^{T}\int_{0}^{B_{s}}\left(rK-\delta S_{s}\right)e^{-r\left(s-t\right)}\psi\left(S_{s};S_{t}=B_{t}\right)dS_{s}ds.\label{eq:ao_eep2}
\end{equation}
\end{theorem}
The equations (\ref{eq:ao_eo})-(\ref{eq:ao_bdy}) in Theorem \ref{thm:eep} for the valuation
of American options can be simplified if we make further assumptions
on the model. For example, if we assume that the stock price follows a geometric
Brownian Motion (GBM), that is, $r\left(S_{t};\theta\right)=r$, $\sigma\left(S_{t};\theta\right)=\sigma$
for constants $r$ and $\sigma$, and $\delta\left(S_{t};\theta\right)=0$,
then we have a Black-Scholes style formula for the valuation of American
options. We summarize this result in the following lemma.

\begin{lemma}[Exercise Premium Representation under GBM]\label{lemma:gbm}
If the stock price
$S$ follows geometric Brownian Motion, then the value of the American put option,
$P_{0}$, can be written as:
\begin{equation}
P_{0}=Ke^{-rT}N\left(k_{2}\left(S_{0},K,T\right)\right)-S_{0}N\left(k_{1}\left(S_{0},K,T\right)\right)+rk\int_{0}^{T}e^{-rt}N\left(b_{2}\left(S_{0},B_{t},t\right)\right)dt,\label{eq:ao_gbm}
\end{equation}
where
\[
k_{1}\left(S_{0},K,T\right)\equiv\frac{\log\left(K/S_{0}\right)-\rho_{1}T}{\sigma\sqrt{T}},
\]
\[
k_{2}\left(S_{0},K,T\right)\equiv\frac{\log\left(K/S_{0}\right)-\rho_{2}T}{\sigma\sqrt{T}},
\]
\[
\rho_{1}\equiv\rho_{2}+\sigma^{2}=r+\frac{\sigma^{2}}{2},
\]
\[
b_{2}\left(S_{0},B_{t},t\right)\equiv\frac{\log\left(B_{t}/S_{0}\right)-\rho_{2}t}{\sigma\sqrt{t}},
\]
and $B_{t}$ solves the following integral equation:
\begin{multline}
K-B_{t}=Ke^{-r\left(T-t\right)}N\left(k_{2}\left(B_{t},K,T-t\right)\right)-B_{t}N\left(k_{1}\left(B_{t},K,T-t\right)\right)\\
+rk\int_{t}^{T}e^{-r\left(s-t\right)}N\left(b_{2}\left(B_{t},B_{s},s-t\right)\right)ds.\label{eq:ao_gbm_bdy}
\end{multline}
\end{lemma}

\subsection{Hermite Polynomial-based Approximation}\label{subsection:hp}

Theorem \ref{thm:eep} provides an intuitive approach to the valuation of American
options in diffusion models; however, we still have two difficulties.
First, most of the diffusion models do not admit a closed-form solution
for the transition density. Second, the exercise boundary $\mathcal{B}$
is unknown in (\ref{eq:ao_eep}), and we need to solve the integral
equation (\ref{eq:ao_bdy}) recursively to compute $\mathcal{B}$.

In this study, we propose the use of the Hermite polynomials to approximate
the transition density. Our approach is based on the work of Ait-Sahalia (2002, 2006) and
Yu (2007). The Hermite polynomial approach by Ait-Sahalia (2002) provided
an explicit sequence of closed-form functions to approximate the unknown
transition density. Ait-Sahalia (2006) and Yu (2007) extended the approach
to multivariate case and jump-diffusion models.

To approximate the transition density of the stock price $S$, we
first transform $S$ into a new random variable $Y$ by defining $Y\equiv\gamma\left(S\right)=\int^{S}du/\sigma\left(u\right)$.
We know that $Y$ has a unit diffusion, that is,
\[
dY_{t}=\mu_{Y}\left(Y_{t};\theta\right)dt+dW_{t},
\]
where
\begin{equation}
\mu_{Y}\left(y;\theta\right)=\frac{\mu\left(\gamma^{-1}\left(y;\theta\right);\theta\right)}{\sigma\left(\gamma^{-1}\left(y;\theta\right);\theta\right)}-\frac{1}{2}\frac{\partial\sigma}{\partial S}\left(\gamma^{-1}\left(y;\theta\right);\theta\right).\label{eq:herm_miuy}
\end{equation}
We denote the domain of $Y$ as $D_{Y}=\left(\underline{y},\bar{y}\right)$.
According to Ait-Sahalia (2002), the transition density of $Y$ can
be approximated using Hermite polynomials, and the transition
density of $S$ can then be derived from that of $Y.$ More specifically,
the transition density of $S$ with time interval $\varDelta$ can be
approximated up to order $m$ as following:
\begin{multline}
\tilde{\psi}^{\left(m\right)}\left(S_{t+\varDelta}=S^{\prime};S_{t}=S\right)=\sigma^{-1}\left(S^{\prime};\theta\right)\varDelta^{-\frac{1}{2}}\phi\left(\frac{\gamma\left(S^{\prime};\theta\right)-\gamma\left(S;\theta\right)}{\varDelta^{\frac{1}{2}}}\right)\times\\
\exp\left(\int_{\gamma\left(S;\theta\right)}^{\gamma\left(S^{\prime};\theta\right)}\mu\left(w;\theta\right)dw\right)\times\sum_{k=0}^{m}c_{k}\left(\gamma\left(S^{\prime};\theta\right)|\gamma\left(S;\theta\right);\theta\right)\frac{\varDelta^{k}}{k!},\label{eq:herm_density}
\end{multline}
where $\phi\left(z\right)\equiv\exp\left(-z^{2}/2\right)/\sqrt{2\pi}$
denotes the density function of standard normal distribution, and
for all $j\geq1$,
\begin{multline}
c_{j}\left(\gamma\left(S^{\prime};\theta\right)|\gamma\left(S;\theta\right);\theta\right)=j\left(S^{\prime}-S\right)^{-j}\int_{\gamma\left(S;\theta\right)}^{\gamma\left(S^{\prime};\theta\right)}\left(w-\gamma\left(S;\theta\right)\right)^{j-1}\\
\times\left\{ \lambda\left(w;\theta\right)c_{j-1}\left(w|\gamma\left(S;\theta\right);\theta\right)+\left(\partial^{2}c_{j-1}\left(w|\gamma\left(S;\theta\right);\theta\right)/\partial w^{2}\right)/2\right\} dw\label{eq:herm_c}
\end{multline}
where $\lambda\left(x;\theta\right)\equiv-\left(\mu_{Y}^{2}\left(x;\theta\right)+\partial\mu_{Y}\left(x;\theta\right)/\partial x\right)/2$
with $\mu_{Y}$ defined in (\ref{eq:herm_miuy}), and $c_{0}=1.$

Once we obtain the approximation of the transition density of $S$ in
(\ref{eq:herm_density}), we can plug $\tilde{\psi}^{\left(m\right)}$
into Theorem \ref{thm:eep} and obtain the approximation of the valuation of American
options. More specifically, we have the following approximated early
exercise premium representation for the value of the American put option up to
order $m$:
\begin{equation}
\tilde{P}_{0}^{\left(m\right)}=\tilde{p}_{0}^{\left(m\right)}+\tilde{e}_{0}^{\left(m\right)}\label{eq:ao_appr}
\end{equation}
where $\tilde{p}_{0}^{\left(m\right)}\equiv\tilde{p}^{\left(m\right)}\left(0,S_{0}=s_{0}\right)=\int_{0}^{K}e^{-rT}\left(K-S_{T}\right)\tilde{\psi}^{\left(m\right)}\left(S_{T};S_{0}=s_{0}\right)dS_{T}$
represents the approximated price of a European put option and $\tilde{e}_{0}^{\left(m\right)}$
is the approximated early exercise premium given by
\begin{equation}
\tilde{e}_{0}^{\left(m\right)}\equiv\tilde{e}^{\left(m\right)}\left(0,S_{0}=s_{0},\tilde{B}^{\left(m\right)}\left(\cdot\right)\right)=\int_{0}^{T}\int_{0}^{\tilde{B}_{t}^{\left(m\right)}}\left(rK-\delta S_{t}\right)e^{-rt}\tilde{\psi}^{\left(m\right)}\left(S_{t};S_{0}=s_{0}\right)dS_{t}dt.\label{eq:ao_eep_appr}
\end{equation}
The approximated exercise boundary up to order $m$, $\tilde{B}_{t}^{\left(m\right)}$,
solves the following recursive nonlinear integral equation:
\begin{equation}
K-\tilde{B}_{t}^{\left(m\right)}=\tilde{p}^{\left(m\right)}\left(t,\tilde{B}_{t}^{\left(m\right)}\right)+\tilde{e}^{\left(m\right)}\left(t,\tilde{B}_{t}^{\left(m\right)},\tilde{B}^{\left(m\right)}\left(\cdot\right)\right)\quad\forall t\in\left[0,T\right).\label{eq:ao_bdy_appr}
\end{equation}
Similarly to (\ref{eq:ao_eo2})--(\ref{eq:ao_eep2}), $\tilde{p}^{\left(m\right)}$
and $\tilde{e}^{\left(m\right)}$ are defined as:
\begin{equation}
\tilde{p}^{\left(m\right)}\left(t,\tilde{B}_{t}^{\left(m\right)}\right)\equiv\int_{0}^{K}\left(K-S_{T}\right)\tilde{\psi}^{\left(m\right)}\left(S_{T};S_{t}=\tilde{B}_{t}^{\left(m\right)}\right)dS_{T},\label{eq:ao_eo2_appr}
\end{equation}
\begin{equation}
\tilde{e}^{\left(m\right)}\left(t,\tilde{B}_{t}^{\left(m\right)},\tilde{B}^{\left(m\right)}\left(\cdot\right)\right)\equiv\int_{t}^{T}\int_{0}^{\tilde{B}_{s}^{\left(m\right)}}\left(rK-\delta S_{s}\right)e^{-r\left(s-t\right)}\tilde{\psi}^{\left(m\right)}\left(S_{s};S_{t}=\tilde{B}_{t}^{\left(m\right)}\right)dS_{s}ds,\label{eq:ao_eep2_appr}
\end{equation}
subject to the boundary condition 
\[
\tilde{B}_{T-}^{\left(m\right)}\equiv\lim_{t\uparrow T}\tilde{B}_{t}^{\left(m\right)}=\min\left\{ K,\frac{r\left(B_{T};\theta\right)}{\delta\left(B_{T};\theta\right)}K\right\} ,
\]
and $\tilde{B}_{T}^{\left(m\right)}=B_{T}=K\geq\tilde{B}_{T-}^{\left(m\right)}$.

The following theorem guarantees that the proposed approach in (\ref{eq:ao_appr})--(\ref{eq:ao_eep2_appr})
is a well-behaved approximation of the value of American options.
\newpage
\begin{theorem}\label{thm:aoh}
Under Assumptions \ref{ass_ao_1}--\ref{ass_ao_3} given in Appendix A, as $m\rightarrow\infty,$
we have
\begin{enumerate}
\item $\tilde{p}_{0}^{\left(m\right)}\rightarrow p_{0}$,
\item $\tilde{B}_{t}^{\left(m\right)}\rightarrow B_{t}$ for any $t\in\left[0,T\right]$,
\item $\tilde{e}_{0}^{\left(m\right)}\rightarrow e_{0}$,
\item $\tilde{P}_{0}^{\left(m\right)}\rightarrow P_{0}$.
\end{enumerate}
\end{theorem}

\begin{proof}
\textit{In Appendix B}
\end{proof}

\section{Valuation of American Options in Jump-Diffusion Models}\label{section:ao_3}

\subsection{Valuation of American Options}

In this section, we discuss the approximation of the value of American
options when the underlying asset price follows a jump-diffusion process.
Because of the discontinuous nature of the asset price path, the exercise
premium representation is different from that without jumps. Specifically, we consider the stock price under the risk-neutral measure, and assume that it follows:
\begin{equation}
d\log S_{t}=\left(r\left(S_{t};\theta\right)-\delta\left(S_{t};\theta\right)-\rho j\right)dt+\sigma\left(S_{t};\theta\right) dW_{t}+\left(J-1\right)dq_{t}\label{eq:aomodel_jump}
\end{equation}
where $dq$ is a Poisson process with rate $\rho t$, $J-1$ is the
proportional change in the price due to a jump with density function $\nu$ as a function of jump size with support $D_{J}$, and $j=E\left(J-1\right)$. We assume $r\left(S_{t};\theta\right)$, $\delta\left(S_{t};\theta\right)$, and $\sigma\left(S_{t};\theta\right)$ are smooth functions of $S_{t}$.
Then, based on Gukhal (2001), the value of the American put option, $P_{0}\equiv P\left(0,S_{0}=s_{0}\right)$, has the following
representation:
\begin{equation}
P_{0}=p_{0}+e_{0}+g_{0}\label{eq:ao_jump}
\end{equation}
where
\begin{equation}
p_{0}\equiv p\left(0,S_{0}=s_{0}\right)=\int_{0}^{K}e^{-rT}\left(K-S_{T}\right)\psi\left(S_{T};S_{0}=s_{0}\right)dS_{T}\label{eq:ao_jump_eo}
\end{equation}
\begin{equation}
e_{0}\equiv e\left(0,S_{0}=s_{0},B\left(\cdot\right)\right)=\int_{0}^{T}\int_{0}^{B_{t}}\left(rK-\delta S_{t}\right)e^{-rt}\psi\left(S_{t};S_{0}=s_{0}\right)dS_{t}dt\label{eq:ao_jump_eep}
\end{equation}
and
\begin{multline}
g_{0}\equiv g\left(0,S_{0}=s_{0},B\left(\cdot\right)\right)=\int_{0}^{T}\int_{0}^{B_{t-}}\int_{B_{t}}^{\infty}e^{-rt}\rho\left(P\left(t,J_{t}S_{t-}\right)-\left(K-J_{t}S_{t-}\right)\right)\\
\times\psi\left(S_{t-};S_{0}=s_{0}\right)\psi\left(J_{t}S_{t-};S_{t-}\right)d\left(J_{t}S_{t-}\right)dS_{t-}dt.\label{eq:ao_jump_jump}
\end{multline}
The exercise boundary $B_{t}$ solves the following integral equation
\begin{equation}
K-B_{t}=p\left(t,B_{t}\right)+e\left(t,B_{t},B\left(\cdot\right)\right)-g\left(t,B_{t},B\left(\cdot\right)\right)\label{eq:ao_jump_bdy}
\end{equation}
where
\[
p\left(t,B_{t}\right)=\int_{0}^{K}\left(K-S_{T}\right)\psi\left(S_{T};B_{t}\right)dS_{T}
\]
\[
e\left(t,B_{t},B\left(\cdot\right)\right)=\int_{t}^{T}\int_{0}^{B_{s}}\left(rK-\delta S_{s}\right)e^{-r\left(s-t\right)}\psi\left(S_{s};B_{t}\right)dS_{s}ds
\]
\begin{multline*}
g\left(t,B_{t},B\left(\cdot\right)\right)=\int_{t}^{T}\int_{0}^{B_{s-}}\int_{B_{s}}^{\infty}e^{-r\left(s-t\right)}\rho\left(P\left(s,J_{s}S_{s-}\right)-\left(K-J_{s}S_{s-}\right)\right)\\
\times\psi\left(S_{s-};B_{t}\right)\psi\left(J_{s}S_{s-};S_{s-}\right)d\left(J_{s}S_{s-}\right)dS_{s-}ds.
\end{multline*}

The representation in (\ref{eq:ao_jump}) has a straightforward interpretation.
As in the case without jumps, $p_{0}$ represents the price of a European
put option, $e_{0}$ is the early exercise premium, and $g_{0}$ is the rebalancing
cost due to the jumps of stock prices from the exercise region (the stock price is below
the exercise boundary) into the continuation region (the stock price
is above the exercise boundary).

\subsection{Hermite Polynomial-based Approximation}

Our approach to studying jump-diffusion models is similar to our approach in
Section \ref{section:ao_2}. We first approximate the transition density using Hermite
polynomials. According to Yu (2007), an approximation of the order $m>0$
is obtained as follows:
\begin{multline}
\tilde{\psi}^{\left(m\right)}\left(S_{t+\varDelta}=S^{\prime};S_{t}=S\right)=\varDelta^{-\frac{1}{2}}\exp\left[-\frac{C^{\left(-1\right)}\left(S,S^{\prime}\right)}{\varDelta}\right]\sum_{k=0}^{m}C^{\left(k\right)}\left(S,S^{\prime}\right)\varDelta^{k}\\
+\sum_{k=1}^{m}D^{\left(k\right)}\left(S,S^{\prime}\right)\varDelta^{k}\label{eq:herm_jump_density}
\end{multline}
where
\begin{equation}
C^{\left(-1\right)}\left(S,S^{\prime}\right)=\frac{1}{2}\left[\int_{S}^{S^{\prime}}\sigma\left(s\right)^{-1}ds\right]^{2},\label{eq:herm_jump_c-1}
\end{equation}
\begin{equation}
C^{\left(0\right)}\left(S,S^{\prime}\right)=\frac{1}{\sqrt{2\pi}\sigma\left(S^{\prime}\right)}\exp\left[\int_{S}^{S^{\prime}}\frac{\mu\left(s\right)}{\sigma^{2}\left(s\right)}-\frac{\sigma^{\prime}\left(s\right)}{2\sigma\left(s\right)}ds\right],\label{eq:herm_jump_c0}
\end{equation}
\begin{multline}
C^{\left(k+1\right)}\left(S,S^{\prime}\right)=-\left[\int_{S}^{S^{\prime}}\sigma\left(s\right)^{-1}ds\right]^{-k+1}\int_{S}^{S^{\prime}}\left\{\exp\left[\int_{u}^{S}\frac{\mu\left(u\right)}{\sigma^{2}\left(u\right)}-\frac{\sigma^{\prime}\left(u\right)}{2\sigma\left(u\right)}du\right]\right.\\
\left.\times\sigma\left(s\right)^{-1}\left[\int_{s}^{S^{\prime}}\sigma\left(u\right)^{-1}du\right]^{k}\left[\rho\left(s\right)-\mathcal{L}\right]C^{\left(k\right)}\left(s,S^{\prime}\right)\right\}ds,\quad for\;k\geq0,\label{eq:herm_jump_ck}
\end{multline}
\begin{equation}
D^{\left(1\right)}\left(S,S^{\prime}\right)=\rho\left(S\right)-\upsilon\left(S^{\prime}-S\right),\label{eq:herm_jump_d1}
\end{equation}
\begin{multline}
D^{\left(k+1\right)}\left(S,S^{\prime}\right)=\frac{1}{1+k}\left[\mathfrak{L}D^{\left(k\right)}\left(S,S^{\prime}\right)+\right.\\
\left.\sqrt{2\pi}\rho\left(S\right)\sum_{r=0}^{k}\frac{M_{2r}^{1}}{\left(2r\right)!}\frac{\partial^{2r}}{\partial w^{2r}}m_{k-r}\left(S,S^{\prime},w\right)|_{w=0}\right],\quad for\;k\geq0,\label{eq:herm_jump_dk}
\end{multline}
where 
\begin{equation}
m_{k}\left(S,S^{\prime},w\right)\equiv C^{\left(k\right)}\left(w_{B}^{-1}\left(w\right),S^{\prime}\right)\upsilon\left(w_{B}^{-1}\left(w\right)-S\right)\sigma\left(w_{B}^{-1}\left(w\right)\right),\label{eq:herm_jump_m}
\end{equation}
\begin{equation}
M_{2r}^{1}\equiv1/\sqrt{2\pi}\int_{\mathbb{R}}\exp\left(-s^{2}/2\right)s^{2r}ds,\label{eq:herm_jump_M}
\end{equation}
\begin{equation}
w_{B}\left(S,S^{\prime}\right)=\int_{S^{\prime}}^{S}\sigma\left(s\right)^{-1}ds,\label{eq:herm_jump_w}
\end{equation}
\begin{equation}
\mathscr{L}f\left(s,s^{\prime}\right)=\frac{1}{2}\sigma^{2}s^{2}\frac{\partial^{2}f}{\partial s^{2}}\left(s,s^{\prime}\right)+\left(r-\delta-\rho j\right)s\frac{\partial f}{\partial s}\left(s,s^{\prime}\right),\label{eq:herm_jump_A}
\end{equation}
and
\begin{equation}
\mathfrak{L}f\left(s,s^{\prime}\right)=\mathscr{L}f\left(s,s^{\prime}\right)+\rho\int_{D_{J}}\left[f\left(s+c,s^{\prime}\right)-f\left(s,s^{\prime}\right)\right]\upsilon\left(c\right)dc.\label{eq:herm_jump_L}
\end{equation}

Once we obtain the Hermite polynomial approximation of the transition
density as above, we plug the approximation into (\ref{eq:ao_jump_eo})-(\ref{eq:ao_jump_jump}),
and solve the integral equation (\ref{eq:ao_jump_bdy}) recursively.
More specifically, we have the approximated value of the American put option up
to order $m$, $\tilde{P}_{0}^{\left(m\right)}\equiv\tilde{P}^{\left(m\right)}\left(0,S_{0}=s_{0}\right)$:
\begin{equation}
\tilde{P}_{0}^{\left(m\right)}=\tilde{p}_{0}^{\left(m\right)}+\tilde{e}_{0}^{\left(m\right)}+\tilde{g}_{0}^{\left(m\right)}\label{eq:ao_jump_appr}
\end{equation}
where
\begin{equation}
\tilde{p}_{0}^{\left(m\right)}\equiv\tilde{p}^{\left(m\right)}\left(0,S_{0}=s_{0}\right)=\int_{0}^{K}e^{-rT}\left(K-S_{T}\right)\tilde{\psi}^{\left(m\right)}\left(S_{T};S_{0}=s_{0}\right)dS_{T},\label{eq:ao_jump_eo_appr}
\end{equation}
\begin{equation}
\tilde{e}_{0}^{\left(m\right)}\equiv\tilde{e}^{\left(m\right)}\left(0,S_{0}=s_{0},\tilde{B}^{\left(m\right)}\left(\cdot\right)\right)=\int_{0}^{T}\int_{0}^{\tilde{B}_{t}^{\left(m\right)}}\left(rK-\delta S_{t}\right)e^{-rt}\tilde{\psi}^{\left(m\right)}\left(S_{t};S_{0}=s_{0}\right)dS_{t}dt,\label{eq:ao_jump_eep_appr}
\end{equation}
and
\begin{multline}
\tilde{g}_{0}^{\left(m\right)}\equiv\tilde{g}^{\left(m\right)}\left(0,S_{0}=s_{0},\tilde{B}^{\left(m\right)}\left(\cdot\right)\right)=\int_{0}^{T}\int_{0}^{\tilde{B}_{t-}^{\left(m\right)}}\int_{\tilde{B}_{t}^{\left(m\right)}}^{\infty}e^{-rt}\rho\left(\tilde{P}^{\left(m\right)}\left(t,J_{t}S_{t-}\right)-\left(K-J_{t}S_{t-}\right)\right)\\
\times\tilde{\psi}^{\left(m\right)}\left(S_{t-};S_{0}=s_{0}\right)\tilde{\psi}^{\left(m\right)}\left(J_{t}S_{t-};S_{t-}\right)d\left(J_{t}S_{t-}\right)dS_{t-}dt.\label{eq:ao_jump_jump_appr}
\end{multline}
The approximated exercise boundary up to order $m$, $\tilde{B}_{t}^{\left(m\right)}$,
solves the following recursive nonlinear integral equation:
\begin{equation}
K-\tilde{B}_{t}^{\left(m\right)}=\tilde{p}^{\left(m\right)}\left(t,\tilde{B}_{t}^{\left(m\right)}\right)+\tilde{e}^{\left(m\right)}\left(t,\tilde{B}_{t}^{\left(m\right)},\tilde{B}^{\left(m\right)}\left(\cdot\right)\right)-\tilde{g}^{\left(m\right)}\left(t,\tilde{B}_{t}^{\left(m\right)},\tilde{B}^{\left(m\right)}\left(\cdot\right)\right),\label{eq:ao_jump_bdy_appr}
\end{equation}
where
\begin{equation}
\tilde{p}^{\left(m\right)}\left(t,\tilde{B}_{t}^{\left(m\right)}\right)\equiv\int_{0}^{K}\left(K-S_{T}\right)\tilde{\psi}^{\left(m\right)}\left(S_{T};S_{t}=\tilde{B}_{t}^{\left(m\right)}\right)dS_{T},\label{eq:ao_jump_eo2_appr}
\end{equation}
\begin{equation}
\tilde{e}^{\left(m\right)}\left(t,\tilde{B}_{t}^{\left(m\right)},\tilde{B}^{\left(m\right)}\left(\cdot\right)\right)\equiv\int_{t}^{T}\int_{0}^{\tilde{B}_{s}^{\left(m\right)}}\left(rK-\delta S_{s}\right)e^{-r\left(s-t\right)}\tilde{\psi}^{\left(m\right)}\left(S_{s};S_{t}=\tilde{B}_{t}^{\left(m\right)}\right)dS_{s}ds,\label{eq:ao_jump_eep2_appr}
\end{equation}
\begin{multline}
\tilde{g}^{\left(m\right)}\left(t,\tilde{B}_{t}^{\left(m\right)},\tilde{B}^{\left(m\right)}\left(\cdot\right)\right)\equiv\int_{t}^{T}\int_{0}^{\tilde{B}_{s-}^{\left(m\right)}}\int_{\tilde{B}_{s}^{\left(m\right)}}^{\infty}e^{-r\left(s-t\right)}\rho\left(\tilde{P}^{\left(m\right)}\left(s,J_{s}S_{s-}\right)-\left(K-J_{s}S_{s-}\right)\right)\\
\times\tilde{\psi}^{\left(m\right)}\left(S_{s-};\tilde{B}_{t}^{\left(m\right)}\right)\tilde{\psi}^{\left(m\right)}\left(J_{s}S_{s-};S_{s-}\right)d\left(J_{s}S_{s-}\right)dS_{s-}ds.\label{eq:ao_jump_jump2_appr}
\end{multline}

\section{Numerical Method and Algorithm}\label{section:ao_4}

Following Detemple (2006), we divide the period $\left[0,T\right]$
into $N$ equal subintervals and let $\Delta=T/N$. We then use a
step function to compute the exercise boundary recursively. The algorithm
works as follows: suppose that our step function approximation of the exercise
boundary is $\left\{ \tilde{B}_{n\Delta}^{\left(m,N\right)},n=0,\ldots,N\right\} $.
 The terminal condition tells us that
 \[
 \tilde{B}_{N\Delta}^{\left(m,N\right)}=\min\left\{ K,\frac{r\left(B_{T};\theta\right)}{\delta\left(B_{T};\theta\right)}\times K\right\}.
 \]
Suppose that $\tilde{B}_{l\Delta}^{\left(m,N\right)}$ is known for
all $l>n$, then $\left\{ \tilde{B}_{l\Delta}^{\left(m,N\right)},l=0,\ldots,n\right\} $
can be obtained by discretizing the integral in (\ref{eq:ao_bdy_appr})
for a diffusion model, or (\ref{eq:ao_jump_bdy_appr}) for a jump-diffusion
model using the trapezoidal rule. For example, we obtain the following
equation for diffusion models:
\begin{multline}
K-\tilde{B}_{l\Delta}^{\left(m,N\right)}=\tilde{p}^{\left(m\right)}\left(l\Delta,\tilde{B}_{l\Delta}^{\left(m,N\right)}\right)+\sum_{q=l+1}^{N-l}\tilde{\epsilon}^{\left(m\right)}\left(\left(q-l\right)\Delta,\tilde{B}_{l\Delta}^{\left(m,N\right)},\tilde{B}_{q\Delta}^{\left(m,N\right)}\right)\Delta\\
+\left[\tilde{\epsilon}^{\left(m\right)}\left(0,\tilde{B}_{l\Delta}^{\left(m,N\right)},\tilde{B}_{l\Delta}^{\left(m,N\right)}\right)+\tilde{\epsilon}^{\left(m\right)}\left(\left(N-l\right)\Delta,\tilde{B}_{l\Delta}^{\left(m,N\right)},\tilde{B}_{N\Delta}^{\left(m,N\right)}\right)\right]\frac{\Delta}{2}\label{eq:ao_num}
\end{multline}
where 
\begin{equation}
\tilde{\epsilon}^{\left(m\right)}\left(s\Delta,\tilde{B}_{t}^{\left(m\right)},\tilde{B}_{t+s\Delta}^{\left(m\right)}\right)\equiv\int_{0}^{\tilde{B}_{t+s\Delta}^{\left(m\right)}}\left(rK-\delta S_{t+s\Delta}\right)e^{-rs\Delta}\tilde{\psi}^{\left(m\right)}\left(S_{t+s\Delta};S_{t}=\tilde{B}_{t}^{\left(m\right)}\right)dS_{t+s\Delta}\label{eq:ao_num_e}
\end{equation}
And for jump-diffusion models, we have
\begin{multline}
K-\tilde{B}_{l\Delta}^{\left(m,N\right)}=\tilde{p}^{\left(m\right)}\left(l\Delta,\tilde{B}_{l\Delta}^{\left(m,N\right)}\right)+\sum_{q=l+1}^{N-l}\tilde{\epsilon}^{\left(m\right)}\left(\left(q-l\right)\Delta,\tilde{B}_{l\Delta}^{\left(m,N\right)},\tilde{B}_{q\Delta}^{\left(m,N\right)}\right)\Delta\\
+\left[\tilde{\epsilon}^{\left(m\right)}\left(0,\tilde{B}_{l\Delta}^{\left(m,N\right)},\tilde{B}_{l\Delta}^{\left(m,N\right)}\right)+\tilde{\epsilon}^{\left(m\right)}\left(\left(N-l\right)\Delta,\tilde{B}_{l\Delta}^{\left(m,N\right)},\tilde{B}_{N\Delta}^{\left(m,N\right)}\right)\right]\frac{\Delta}{2}\\
+\left[\tilde{\eta}^{\left(m\right)}\left(0,\tilde{B}_{l\Delta}^{\left(m,N\right)},\tilde{B}_{l\Delta}^{\left(m,N\right)}\right)+\tilde{\eta}^{\left(m\right)}\left(\left(N-l\right)\Delta,\tilde{B}_{l\Delta}^{\left(m,N\right)},\tilde{B}_{N\Delta}^{\left(m,N\right)}\right)\right]\frac{\Delta}{2}\\
+\sum_{q=l+1}^{N-l}\tilde{\eta}^{\left(m\right)}\left(\left(q-l\right)\Delta,\tilde{B}_{l\Delta}^{\left(m,N\right)},\tilde{B}_{q\Delta}^{\left(m,N\right)}\right)\Delta\label{eq:ao_jump_num}
\end{multline}
where $\tilde{\epsilon}^{\left(m\right)}$ is defined as in (\ref{eq:ao_num_e}),
and $\tilde{\eta}^{\left(m\right)}$ is defined by
\begin{multline}
\tilde{\eta}^{\left(m\right)}\left(s\Delta,\tilde{B}_{t}^{\left(m\right)},\tilde{B}_{t+s\Delta}^{\left(m\right)}\right)\equiv\\
\int_{0}^{\tilde{B}_{t+s\Delta-}^{\left(m\right)}}\int_{\tilde{B}_{t+s\Delta}^{\left(m\right)}}^{\infty}e^{-rs\Delta}\rho\left(\tilde{P}^{\left(m\right)}\left(s\Delta,J_{t+s\Delta}S_{t+s\Delta s-}\right)-\left(K-J_{t+s\Delta}S_{t+s\Delta s-}\right)\right)\\
\times\tilde{\psi}^{\left(m\right)}\left(S_{t+s\Delta-};\tilde{B}_{t}^{\left(m\right)}\right)\tilde{\psi}^{\left(m\right)}\left(J_{t+s\Delta}S_{t+s\Delta-};S_{t+s\Delta-}\right)d\left(J_{t+s\Delta}S_{t+s\Delta-}\right)dS_{t+s\Delta-}\label{eq:ao_num_g}
\end{multline}

We run the above procedure recursively, and obtain the exercise
boundary $\left\{ \tilde{B}_{n\Delta}^{\left(m,N\right)},n=0,\ldots,N\right\} $.
Finally, the value of the American put option can be computed by substituting the exercise
boundary into (\ref{eq:ao_appr}) for diffusion models or (\ref{eq:ao_jump_appr})
for jump-diffusion models.

\section{Applications}\label{section:ao_5}

In this section, we illustrate how to compute the value of the American put option and the corresponding exercise boundary for diffusion models
and jump-diffusion models. We use these examples to illustrate
the accuracy and speed of the proposed approach. We approximate the transition density using $m=2$ for all the examples in this section.

\subsection{Applications to Diffusion Models}

\subsubsection{Geometric Brownian Motion (GBM) Model}

In the geometric Brownian Motion model, the stock price $S$ follows
\[
dS_{t}=\left(r-\delta\right)S_{t}dt+\sigma S_{t}dW_{t},
\]
 where $r$, $\delta$, and $\sigma$ are constants. To test the efficiency of our recursive algorithm in Section \ref{section:ao_2}, we
compare the results from our approach with those
from four widely used methods: the binomial method by Cox, Ross, and
Rubinstein (1979), the accelerated binomial methods by Breen (1991), the
finite difference method, and the analytical approximation by Geske
and Johnson (1984). We use the results from the binomial method with
10,000 time-steps as a benchmark to measure the accuracy. Following
Huang \emph{et al.} (1996) and Geske and Johnson (1984), we
set $S_{0}=40$, $r=4.88\%$, and $\delta=0$.

Table \ref{gbm} reports the valuations of American options from the six approaches. Columns 1 through 3 represent the values of the parameters, $K$ (strike price), $\sigma$ (volatility), and $T$ (maturity), respectively.
Column 4 gives the numerical results from the binomial method with
10,000 time-steps, and we take this approach as a benchmark. Column
5 includes the results in Table I of Geske and Johnson (1984). Columns
6 through 8 report the results from the binomial method with 150 time-steps, the finite difference method with 200 steps, and the accelerated
binomial method with 150 time-steps. Column 9 shows the results
of the proposed approach with 100 time-steps. The accuracy is measured by the root mean squared error, as shown in the last row. It is clear from this table that the proposed approach achieves the best performance in terms of accuracy compared with the other methods.

\begin{table}[H]
\centering
\caption{Value of American Options Based on Different Numerical Methods}
\renewcommand{\arraystretch}{1.2}
\begin{threeparttable}
\begin{tabular}{|c|c|c|c|c|c|c|c|c|}
\hline 
\emph{K} & $\sigma$ & \emph{T }(yr) & Binomial & G\&J & Binomial II & Accelerated & FD & Hermite\tabularnewline
\hline 
\hline 
35 & 0.2 & 0.0833 & 0.0062 & 0.0062 & 0.0061 & 0.0061 & 0.0278 & 0.0062\tabularnewline
\hline 
35 & 0.2 & 0.3333 & 0.2004 & 0.1999 & 0.1995 & 0.1994 & 0.2382 & 0.2004\tabularnewline
\hline 
35 & 0.2 & 0.5833 & 0.4328 & 0.4321 & 0.4340 & 0.4331 & 0.4624 & 0.4329\tabularnewline
\hline 
40 & 0.2 & 0.0833 & 0.8522 & 0.8528 & 0.8512 & 0.8517 & 0.9874 & 0.8523\tabularnewline
\hline 
40 & 0.2 & 0.3333 & 1.5798 & 1.5807 & 1.5783 & 1.5752 & 1.6244 & 1.5800\tabularnewline
\hline 
40 & 0.2 & 0.5833 & 1.9904 & 1.9905 & 1.9886 & 1.9856 & 2.0177 & 1.9906\tabularnewline
\hline 
45 & 0.2 & 0.0833 & 5.0000 & 4.9985 & 5.0000 & 4.9200 & 5.0052 & 5.0000\tabularnewline
\hline 
45 & 0.2 & 0.3333 & 5.0883 & 5.0951 & 5.0886 & 4.9253 & 5.1327 & 5.0886\tabularnewline
\hline 
45 & 0.2 & 0.5833 & 5.2670 & 5.2719 & 5.2677 & 5.2844 & 5.2699 & 5.2673\tabularnewline
\hline 
35 & 0.3 & 0.0833 & 0.0774 & 0.0744 & 0.0775 & 0.0772 & 0.1216 & 0.0774\tabularnewline
\hline 
35 & 0.3 & 0.3333 & 0.6975 & 0.6969 & 0.6993 & 0.6977 & 0.7300 & 0.6976\tabularnewline
\hline 
35 & 0.3 & 0.5833 & 1.2198 & 1.2194 & 1.2239 & 1.2218 & 1.2407 & 1.2199\tabularnewline
\hline 
40 & 0.3 & 0.0833 & 1.3099 & 1.3100 & 1.3083 & 1.3095 & 1.3860 & 1.3100\tabularnewline
\hline 
40 & 0.3 & 0.3333 & 2.4825 & 2.4817 & 2.4799 & 2.4781 & 2.5068 & 2.4828\tabularnewline
\hline 
40 & 0.3 & 0.5833 & 3.1696 & 3.1733 & 3.1665 & 3.1622 & 3.1819 & 3.1699\tabularnewline
\hline 
45 & 0.3 & 0.0833 & 5.0597 & 5.0599 & 5.0600 & 5.0632 & 5.1016 & 5.0598\tabularnewline
\hline 
45 & 0.3 & 0.3333 & 5.7056 & 5.7012 & 5.7065 & 5.6978 & 5.7193 & 5.7059\tabularnewline
\hline 
45 & 0.3 & 0.5833 & 6.2436 & 6.2365 & 6.2448 & 6.2395 & 6.2477 & 6.2440\tabularnewline
\hline 
35 & 0.4 & 0.0833 & 0.2466 & 0.2466 & 0.2454 & 0.2456 & 0.2949 & 0.2466\tabularnewline
\hline 
35 & 0.4 & 0.3333 & 1.3460 & 1.3450 & 1.3505 & 1.3481 & 1.3696 & 1.3461\tabularnewline
\hline 
35 & 0.4 & 0.5833 & 2.1549 & 2.1568 & 2.1602 & 2.1569 & 2.1676 & 2.1551\tabularnewline
\hline 
40 & 0.4 & 0.0833 & 1.7681 & 1.7679 & 1.7658 & 1.7674 & 1.8198 & 1.7683\tabularnewline
\hline 
40 & 0.4 & 0.3333 & 3.3874 & 3.3632 & 3.3835 & 3.3863 & 3.4011 & 3.3877\tabularnewline
\hline 
40 & 0.4 & 0.5833 & 4.3526 & 4.3556 & 4.3480 & 4.3426 & 4.3567 & 4.3530\tabularnewline
\hline 
45 & 0.4 & 0.0833 & 5.2868 & 5.2855 & 5.2875 & 5.2863 & 5.3289 & 5.2870\tabularnewline
\hline 
45 & 0.4 & 0.3333 & 6.5099 & 6.5093 & 6.5103 & 6.5054 & 6.5147 & 6.5101\tabularnewline
\hline 
45 & 0.4 & 0.5833 & 7.3830 & 7.3831 & 7.3897 & 7.3785 & 7.3792 & 7.3833\tabularnewline
\hline 
\multicolumn{3}{|c|}{RMSE} & 0.0000 & 5.34e-03 & 2.64e-03 & 3.53e-02 & 4.10e-02 & 2.16e-04\tabularnewline
\hline 
\end{tabular}

\end{threeparttable}

\label{gbm}
\end{table}

We report the approximated exercise boundary in Figure \ref{boundarygbm1} for several combinations of strike prices and volatility. The parameters values are the same as those in Table \ref{gbm}, and $T$ is 0.5833. This figure shows the marginal effect of strike prices and volatility on the exercise boundary. For example, as the volatility becomes smaller, the responding exercise boundary becomes flatter. The intuition for this result is that when the volatility is small, the return from withholding American options is limited; thus, the American put option will be exercised at a higher boundary instead of a lower one. This result can also be confirmed by checking the partial difference of the exercise boundary with respect to the volatility in (\ref{eq:ao_gbm_bdy}).

The GBM model is one of the limited cases in which we know the true transition density, and to examine the accuracy of our approximated exercise boundary, we plug in the true transition density into the numerical algorithm in Section \ref{section:ao_4}, and compare the results with our approximated exercise boundary. We report this comparison for different strike prices and volatility in Figure \ref{boundarygbm2}.

In Figure \ref{fdvsh}, we compare the results of our proposed approach with the finite difference method for approximating the exercise boundary.

To further investigate the performance of our approach, we report the approximated value of the American put option with respect to strike prices from 10 to 70 in Figure \ref{strike}. In addition, we computed the approximated value of the American put option for different strikes based on various orders of the Hermite polynomial-based approximation of the transition density. More specifically, the first order approximation means $m=1$; the second order means $m=2$; the third order means $m=3$. We use the results from the binomial method as a benchmark for comparison, and report the relative error of the approximation in Figure \ref{strikeerror}.

\subsubsection{Constant Elasticity Volatility (CEV) Model}

The Constant Elasticity Volatility model assumes that the stock price
$S$ follows
\[
dS_{t}=\left(r-\delta\right)S_{t}dt+\sigma S_{t}^{\alpha/2}dW_{t},
\]
where $r$, $\delta$, $\sigma$, and $\alpha$ are constants. Detemple
and Kitapbayev (2018) applied this model to study the pricing of the American
VIX option. Further extension of this model on the valuation of the VIX
option can be found in Goard and Mazur (2013).

We set $K=100$, $r=6/100$, $\delta=r/2$, $\sigma=\sqrt{10}/5$, $S_{0}=40$, and $T=1$. We report in Figure \ref{boundarycev} the approximated exercise boundary of the CEV model for $\alpha=1.9$, and $\alpha=1.7$, respectively. From Figure \ref{boundarycev}, we find that as $\alpha$ decreases, the volatility of the stock price decreases, and thus the optimal exercise boundary becomes higher. This result is the same as that found in the GBM model.

\subsubsection{Nonlinear Mean Reversion (NMR) Model}

The Nonlinear Mean Reversion model assumes that 
\[
dS_{t}=\left(\frac{a}{S_{t}}+b+cS_{t}+vS_{t}^{2}\right)dt+\sigma S_{t}^{\gamma}dW_{t},
\]
where $a$, $b$, $c$, $v$, $\sigma$, and $\gamma$ are constants.
This model was discussed in Ait-Sahalia (1996, 1999), and Gallant
and Tauchen (1998) for modeling the interest rates. Eraker and Wang (2012)
proposed a similar model for the VIX option.

In the NMR model, we set $a=500$, $b=5$, $c=0.05$, $v=-0.05$, $\sigma=0.2$, $\gamma=3/2$, $K=20$, $r=5/100$, $\delta=0$, $S_{0}=20$, and $T=0.0833$. We report the approximated exercise boundary shown in Figure \ref{nmr2}.

\subsubsection{Double Mean Reversion (DMR) Model}

The Double Mean Reversion model assumes that 
\[
dS_{t}=\beta\left(y_{t}-S_{t}\right)dt+\sigma\sqrt{S_{t}}dW_{t}
\]
\[
dy_{t}=\xi\left(\alpha-y_{t}\right)dt+\kappa\sqrt{y_{t}}dU_{t}
\]
where $W$ and $U$ are two independent Brownian motions, and $\alpha$,
$\beta$, $\xi$, $\kappa$, and $\sigma$ are constants.

Based on the usual square root model, this DMR model includes an additional
stochastic factor for the mean level of the stock price. In this model,
the speed of mean-reversion towards the short-run stochastic mean
level of the stock price is controlled by $\beta$, and the speed of
mean-reversion towards the long-run mean level of the short-run stochastic
mean is controlled by $\xi$. This model was discussed in Amengual (2008),
Mencia and Sentana (2009), and Egloff \emph{et al.} (2010).

In the DMR model, the optimal exercise boundary is a function of time, $t$, and $y$. We set $K=40/100$, $r=4.88/100$, $\delta=0$, $\sigma=0.25$, $\kappa=0.2$, $\beta=2.5$, $\xi=4$, $\alpha=0.25$, $T=0.5$.
In Figure \ref{dmr}, we report the approximated exercise boundary in the DMR model. The boundary is approximated with 20 steps on time, and 100 steps on $y$ for $y$ in $[0,1]$.

\subsection{Applications to Jump-Diffusion Models}

\subsubsection{Merton's Jump-Diffusion Model}

Merton (1976) proposed a jump-diffusion model to incorporate discontinuous
returns, and derived a closed-form vanilla option pricing formula. Merton's jump-diffusion model assumes that:
\[
d\log S_{t}=\left(r-\delta-\lambda j\right)dt+\sigma dW_{t}+\left(J-1\right)dq_{t},
\]
where $dq$ is a Poisson process with rate $\lambda t$, $J$ has a lognormal
distribution with mean $\mu_{J}$ and variance $\sigma_{J}^{2}$,
and $j=E\left[J-1\right]=\exp\left(\mu_{J}+\sigma_{J}^{2}/2\right)-1$.

We set $K=40$, $r=4.88/100$, $\sigma=0.2$, $\mu_{J}=0$, $\sigma_{J}=0.2$, $S_{0}=40$, and $T=0.5$. Additionally, we approximate the exercise boundary for different values of $\lambda$. More specifically, we try $\lambda=1/100, 10/100,$ and $25/100$, and present the results in Figure \ref{merton}. By comparing the exercise boundaries in Figure \ref{merton}, we find that when $\lambda$ is smaller, the exercise boundary is higher. The intuition for this result is that when $\lambda$ is smaller, the jump in the return occurs less frequently, and thus the return becomes less volatile. Similar to the models without jumps in this section, when the stock price or the return is less volatile, the exercise boundary becomes higher.

\subsubsection{Kou's Jump-Diffusion Model}

To incorporate the leptokurtic feature of the return distribution
and ``volatility smile'' phenomenon in option market, Kou (2002)
proposed a double exponential jump-diffusion model. This model assumes
\[
d\log S_{t}=\left(r-\delta\right)dt+\sigma dW_{t}+Jdq_{t}
\]
where $J$ has an asymmetric double exponential distribution with
density:

\[
\upsilon\left(z\right)=p*\eta_{1}e^{-\eta_{1}z}\mathbf{1}_{\left\{ z\geq0\right\} }+q*\eta_{2}e^{-\eta_{2}z}\mathbf{1}_{\left\{ z<0\right\} },
\]
where $\eta_{1}>1$, $\eta_{2}>0$, $p+q=1$, and $0\leq p,q\leq1$.
The mean, variance, and skewness of the jump size in log returns are:
\[
\varphi_{1}=\frac{p}{\eta_{1}}-\frac{q}{\eta_{2}},
\]
\[
\varphi_{2}=pq\left(\frac{1}{\eta_{1}}+\frac{1}{\eta_{2}}\right)^{2}+\frac{p}{\eta_{1}^{2}}+\frac{q}{\eta_{2}^{2}},
\]
\[
\varphi_{3}=\frac{2\left(p^{3}-1\right)\eta_{1}^{3}-2\left(q^{3}-1\right)\eta_{2}^{3}+6pq\eta_{1}\eta_{2}\left(q\eta_{2}-p\eta_{1}\right)}{\left(p\eta_{2}^{2}+q\eta_{1}^{2}+pq\left(\eta_{1}+\eta_{2}\right)^{2}\right)^{\frac{3}{2}}}.
\]
Also, $dq$ is a Poisson process with rate $\lambda t$.

In Figure \ref{kou}, we report the approximated exercise boundary for Kou's jump-diffusion model with $\lambda=1/100, 10/100,$ and $20/100$, respectively. We set $K=40$, $r=4.88/100$, $\delta=0$, $\sigma=0.2$, $p=0.04$, $q=0.96$, $\eta_{1}=3.7$, $\eta_{2}=1.8$, $S_{0}=40$, and $T=0.5$. We approximate the boundary by 50 steps on time. Similarly, we find that the smaller the intensity of the jump is, the higher the exercise boundary becomes.

\section{Conclusion}\label{section:ao_6}

In this study, we develop a new approach to approximate the exercise boundary and the value of the American put option based on Hermite polynomials. We also provide a numerical scheme for implementing the proposed approach. We show theoretically that our approximation will converge to the true exercise boundary and the value of the American put option, and provide evidence for the efficiency of our approach through several numerical examples including diffusion processes and jump-diffusion processes. We only discuss the case of the American put option; however, the value of the American call option can be approximated similarly.

A drawback of our approach is its computational complexity. Although we have a closed-form approximation of the transition density for a given jump-diffusion model, we need to evaluate the integral of the transition density, and the integral usually does not admit a closed-form solution. This incurs a heavy computational burden on the numerical implementation. Other approaches for approximating the transition density, such as finite mixture models, can simplify the integral equation, and thus reduce the computational complexity. We leave this to be explored in future research.

\newpage
\appendix
\setcounter{table}{0}
\setcounter{figure}{0}
\renewcommand{\thetable}{A\arabic{table}}
\renewcommand{\thefigure}{A\arabic{figure}}

\clearpage
\section{Assumptions}

\begin{assumption}[Smoothness of Coefficients]\label{ass_ao_1}
The functions $r\left(S_{t};\theta\right)$, $\delta\left(S_{t};\theta\right)$
and $\sigma\left(S_{t};\theta\right)$ are infinitely differentiable
in $S$, and three times continuously differentiable in $\theta$,
for all $S\in D_{S}$ and $\theta\in\varTheta$.

\end{assumption}

\begin{assumption}[Non-Degeneracy of the Diffusion]\label{ass_ao_2}
\begin{enumerate}

\item If $D_{S}=\left(-\infty,+\infty\right)$, there exists a constant
$c$ such that $\sigma\left(S_{t};\theta\right)>c>0$ for all $S\in D_{S}$
and $\theta\in\varTheta$.
\item If $D_{S}=\left(0,+\infty\right)$, there exists constants $\zeta_{0}>0$,
$\omega>0$, $\eta\geq0$ such that $\sigma\left(S_{t};\theta\right)\geq\omega S^{\eta}$
for all $0<S\leq\zeta_{0}$ and $\theta\in\varTheta$.
\end{enumerate}

\end{assumption}

\begin{assumption}[Boundary Behavior]\label{ass_ao_3}
For all $\theta\in\varTheta$, $\mu_{Y}\left(y;\theta\right)$ in
(\ref{eq:herm_miuy}) and its derivatives with respect to $y$ and $\theta$
have at most polynomial growth near the boundaries and $\lim_{y\rightarrow\underline{y}^{+}\mathrm{or}\bar{y}^{-}}\lambda\left(y;\theta\right)<+\infty$
where $\lambda\left(y;\theta\right)\equiv-\left(\mu_{Y}^{2}\left(y;\theta\right)+\partial\mu_{Y}\left(y;\theta\right)/\partial y\right)/2$.
\begin{enumerate}
\item Left Boundary: If $\underline{y}=0$, there exist constants $\varepsilon_{0}$,
$\chi$, $\varsigma$ such that for all $0<y\leq\varepsilon_{0}$
and $\theta\in\varTheta$, $\mu_{Y}\left(y;\theta\right)\geq\chi y^{-\varsigma}$
where either $\varsigma>1$ and $\chi>0$, or $\varsigma=1$ and $\chi\geq1$.
If $\underline{y}=-\infty$, there exist constants $E_{0}>0$ and
$K_{0}>0$ such that for all $y\leq-E_{0}$ and $\theta\in\varTheta$,
$\mu_{Y}\left(y;\theta\right)\geq K_{0}y$.
\item Right Boundary: If $\bar{y}=+\infty$, there exist constants $E_{0}>0$
and $K_{0}>0$ such that for all $y\geq E_{0}$ and $\theta\in\varTheta$,
$\mu_{Y}\left(y;\theta\right)\leq K_{0}y$. If $\bar{y}=0$, there
exist constants $\varepsilon_{0}$, $\chi$, $\varsigma$ such that
for all $0>y\geq-\varepsilon_{0}$ and $\theta\in\varTheta$, $\mu_{Y}\left(y;\theta\right)\leq-\chi|y|^{-\varsigma}$
where either $\varsigma>1$ and $\chi>0$, or $\varsigma=1$ and $\chi\geq1/2$.
\end{enumerate}

\end{assumption}

\section{Proof}
\subsection{Proof of Theorem 2}

\textbf{Step 1:} According to Ait-Sahalia (2002), we apply the following transform
to $S$ by $S\rightarrow Y\rightarrow Z$:
\[
Y\equiv\gamma\left(S\right)=\int^{S}du/\sigma\left(u\right),
\]
and
\[
Z\equiv\varDelta^{-\frac{1}{2}}\left(Y-y_{0}\right).
\]
We approximate the transition density of $Z$ by the following Hermite
polynomial construction up to order $m$:
\[
\tilde{\psi}_{Z}^{\left(m\right)}\left(Z_{t+\varDelta}=z^{\prime};Z_{t}=z\right)\equiv\phi\left(z^{\prime}\right)\sum_{j=0}^{m}\eta_{Z}^{\left(j\right)}\left(\varDelta,z;\theta\right)H_{j}\left(z\right),
\]
where $\phi$ is the density of standard normal distribution, $H$
is the Hermite polynomials and $\eta_{Z}$ is the coefficient in the
approximation. We have 
\begin{multline*}
|\eta_{Z}^{\left(j\right)}\left(\varDelta,z;\theta\right)H_{j}\left(z\right)|\leq Q\left\{ 1+|z^{5/2}/2^{5/4}|\right\} e^{z^{2}/4}\times\\
\left\{ j^{-1/2}\left(j+1\right)^{-1}+\left(j+1\right)!v_{j+1}^{2}\left(\varDelta,z\right)\right\} /2,
\end{multline*}
where 
\[
v_{j+1}\left(\varDelta,z\right)=\left(j!\right)^{-1}\int_{-\infty}^{+\infty}H_{j}\left(w\right)\left\{ \frac{\partial p_{Z}\left(\varDelta,w|z;\theta\right)}{\partial w}\right\} dw,
\]
and $Q$ is a constant. $p_{Z}\left(\varDelta,z^{\prime}|z;\theta\right)$
is defined as the true transition density of $Z$. It is easy to verify
that $j^{-1/2}\left(j+1\right)^{-1}\leq\varrho$ (a constant), and
\begin{multline*}
\sum_{j=0}^{m}\left(j\right)!v_{j}^{2}\left(\varDelta,z\right)\leq\int_{-\infty}^{+\infty}e^{w^{2}/2}\left\{ \frac{\partial p_{Z}\left(\varDelta,w|z;\theta\right)}{\partial w}\right\} ^{2}dw\\
\leq\int_{-\infty}^{+\infty}e^{w^{2}/2}\left(b_{0}e^{-3w^{2}/8}R\left(|w|,|z|\right)e^{b_{1}|w||z|+b_{2}|w|+b_{3}|z|+b_{4}z^{2}}\right)dw
\end{multline*}
where $R$ is a polynomial of finite order in $\left(|w|,|z|\right)$
with coefficients uniform in $\theta\in\varTheta$, and where the
constants $b_{i},i=0,\ldots4$, are uniform in $\theta\in\varTheta$.

According to Lebesgue's Dominant Convergence Theorem (DCT), $\sum_{j=0}^{m}\left(j\right)!v_{j}^{2}\left(\varDelta,z\right)$
is convergent, and thus bounded.

Then,
\[
|\tilde{\psi}_{Z}^{\left(m\right)}|\leq\phi\left(z\right)\left\{ Q\left\{ 1+|z^{5/2}/2^{5/4}|\right\} e^{z^{2}/4}\varrho^{\prime}+\sum_{j=0}^{m}\left(j+1\right)!v_{j+1}^{2}\left(\varDelta,z\right)\right\} /2
\]
where $\varrho^{\prime}=m\varrho$.

Notice that $Q\left\{ 1+|z^{5/2}/2^{5/4}|\right\} e^{z^{2}/4}\varrho^{\prime}$
and $\sum_{j=0}^{m}\left(j+1\right)!v_{j+1}^{2}\left(\varDelta,z\right)$
are integrable. It follows from above that $\tilde{\psi}_{Z}^{\left(m\right)}$
is also integrable.

By our assumption, $\sigma$ is globally nondegenerate, that is, there
exists a constant $\xi$ such that $\sigma^{-1}\left(s\right)<\xi^{-1}<\infty$.
We can recover the transition density of $S$ from that of $Z$ by
\[
\tilde{\psi}_{S}^{\left(m\right)}\left(S^{\prime};S\right)=\sigma^{-1}\varDelta^{-1/2}\tilde{\psi}_{Z}^{\left(m\right)}\left(\varDelta^{-1/2}\left(\gamma\left(S^{\prime}\right)-\gamma\left(S\right)\right),\gamma\left(S\right)\right).
\]
It is easy to verify that $\tilde{\psi}_{S}^{\left(m\right)}$ is integrable.
From Ait-Sahalia (2002), we know $\tilde{\psi}_{S}^{\left(m\right)}$ is
convergent to the true transition density $\psi$ as $m\rightarrow\infty$.
Then, by DCT, we have the integral of $\tilde{\psi}_{S}^{\left(m\right)}$
will converge to the integral of $\psi$. That is $\tilde{p}_{0}^{\left(m\right)}\rightarrow p_{0}$
as $m\rightarrow\infty$. We complete the proof for the first part
of Theorem 3.2.

\textbf{Step 2:} By our construction in Subsection 3.2.3 for approximating the exercise
boundary, we have $\tilde{B}_{T}^{\left(m\right)}=B_{T}$ and $\tilde{B}_{T-}^{\left(m\right)}=B_{T-}$.
Now, we assume that $\tilde{B}_{s}^{\left(m\right)}=B_{s}$ for $s>t$,
then $\tilde{e}^{\left(m\right)}\left(t,\tilde{B}_{t}^{\left(m\right)},\tilde{B}^{\left(m\right)}\left(\cdot\right)\right)=\tilde{e}^{\left(m\right)}\left(t,\tilde{B}_{t}^{\left(m\right)},B\left(\cdot\right)\right)$.
Based on our analysis in Step 1, we know that $\tilde{e}^{\left(m\right)}$
and $\tilde{p}^{\left(m\right)}$ are well defined smooth function
of $\tilde{B}_{t}^{\left(m\right)}$.

Let $\tilde{F}^{\left(m\right)}\left(\tilde{B}_{t}^{\left(m\right)}\right)$
be a smooth function of $\tilde{B}_{t}^{\left(m\right)}$ such that
\[
\tilde{F}^{\left(m\right)}\left(\tilde{B}_{t}^{\left(m\right)}\right)\equiv\tilde{p}^{\left(m\right)}\left(t,\tilde{B}_{t}^{\left(m\right)}\right)+\tilde{e}^{\left(m\right)}\left(t,\tilde{B}_{t}^{\left(m\right)},B\left(\cdot\right)\right)+\tilde{B}_{t}^{\left(m\right)},
\]
and let $F\left(B_{t}\right)$ be a smooth function of $B_{t}$ such
that 
\[
F\left(B_{t}\right)\equiv p\left(t,B_{t}\right)+e\left(t,B_{t},B\left(\cdot\right)\right)+B_{t}.
\]
By this construction, we have $\tilde{B}_{t}^{\left(m\right)}=\left(\tilde{F}^{\left(m\right)}\right)^{-1}\left(K\right)$
and $B_{t}=F^{-1}\left(K\right)$ where we denote $\left(\tilde{F}^{\left(m\right)}\right)^{-1}$
as the inverse function of $\tilde{F}^{\left(m\right)}$.

Because $\tilde{\psi}_{S}^{\left(m\right)}\rightarrow\psi$ as $m\rightarrow\infty$,
we know $\tilde{p}^{\left(m\right)}\left(t,\cdot\right)\rightarrow p\left(t,\cdot\right)$
and $\tilde{e}^{\left(m\right)}\left(t,\cdot,B\left(\cdot\right)\right)\rightarrow e\left(t,\cdot,B\left(\cdot\right)\right)$
as $m\rightarrow\infty$. It follows that $\tilde{F}^{\left(m\right)}\rightarrow F$
and so does the inverse. This implies $\left(\tilde{F}^{\left(m\right)}\right)^{-1}\left(K\right)\rightarrow F^{-1}\left(K\right)$,
and thus $\tilde{B}_{t}^{\left(m\right)}\rightarrow B_{t}$ as $m\rightarrow\infty$.

Next, we assume that $\tilde{B}_{s}^{\left(m\right)}\rightarrow B_{s}$
for $s>t$. We denote $\delta_{t}^{\left(m\right)}=\tilde{B}_{t}^{\left(m\right)}-B_{t}$.
It is easy to verify that we still have $\tilde{p}^{\left(m\right)}\left(t,\cdot\right)\rightarrow p\left(t,\cdot\right)$
in this case. For $\tilde{e}^{\left(m\right)}\left(t,\tilde{B}_{t}^{\left(m\right)},\tilde{B}^{\left(m\right)}\left(\cdot\right)\right)$,
we have
\begin{multline*}
\tilde{e}^{\left(m\right)}\left(t,\tilde{B}_{t}^{\left(m\right)},\tilde{B}^{\left(m\right)}\left(\cdot\right)\right)=\int_{t}^{T}\int_{0}^{\tilde{B}_{s}^{\left(m\right)}}\left(rK-\delta S_{s}\right)e^{-r\left(s-t\right)}\tilde{\psi}^{\left(m\right)}\left(S_{s};S_{t}=\tilde{B}_{t}^{\left(m\right)}\right)dS_{s}ds\\
=\int_{t}^{T}\int_{0}^{B_{s}}\left(rK-\delta S_{s}\right)e^{-r\left(s-t\right)}\tilde{\psi}^{\left(m\right)}\left(S_{s};S_{t}=\tilde{B}_{t}^{\left(m\right)}\right)dS_{s}ds\\
+\int_{t}^{T}\int_{0}^{\delta_{t}^{\left(m\right)}}\left(rK-\delta S_{s}\right)e^{-r\left(s-t\right)}\tilde{\psi}^{\left(m\right)}\left(S_{s};S_{t}=\tilde{B}_{t}^{\left(m\right)}\right)dS_{s}ds
\end{multline*}
As $\tilde{\psi}^{\left(m\right)}$ is uniformly integrable, and $\delta_{t}^{\left(m\right)}\rightarrow0$,
we have 
\[
\int_{t}^{T}\int_{0}^{\delta_{t}^{\left(m\right)}}\left(rK-\delta S_{s}\right)e^{-r\left(s-t\right)}\tilde{\psi}^{\left(m\right)}\left(S_{s};S_{t}=\tilde{B}_{t}^{\left(m\right)}\right)dS_{s}ds\rightarrow0
\]
Then it follows that 
\begin{multline*}
\tilde{e}^{\left(m\right)}\left(t,\tilde{B}_{t}^{\left(m\right)},\tilde{B}^{\left(m\right)}\left(\cdot\right)\right)\rightarrow\int_{t}^{T}\int_{0}^{B_{s}}\left(rK-\delta S_{s}\right)e^{-r\left(s-t\right)}\tilde{\psi}^{\left(m\right)}\left(S_{s};S_{t}=\tilde{B}_{t}^{\left(m\right)}\right)dS_{s}ds\\
\rightarrow\int_{t}^{T}\int_{0}^{B_{s}}\left(rK-\delta S_{s}\right)e^{-r\left(s-t\right)}\psi\left(S_{s};S_{t}=\tilde{B}_{t}^{\left(m\right)}\right)dS_{s}ds=e\left(t,\tilde{B}_{t}^{\left(m\right)},B\left(\cdot\right)\right)\\
\end{multline*}

We define 
\[
\tilde{\mathcal{F}}^{\left(m\right)}\left(\tilde{B}_{t}^{\left(m\right)}\right)\equiv\tilde{p}^{\left(m\right)}\left(t,\tilde{B}_{t}^{\left(m\right)}\right)+\tilde{e}^{\left(m\right)}\left(t,\tilde{B}_{t}^{\left(m\right)},\tilde{B}^{\left(m\right)}\left(\cdot\right)\right)+\tilde{B}_{t}^{\left(m\right)}.
\]

Based on the above analysis, we know $\tilde{\mathcal{F}}^{\left(m\right)}\rightarrow F$
as $m\rightarrow\infty$, and thus $\left(\tilde{\mathcal{F}}^{\left(m\right)}\right)^{-1}\left(K\right)$
$\rightarrow F^{-1}\left(K\right)$.
This means $\tilde{B}_{t}^{\left(m\right)}\rightarrow B_{t}$ as $m\rightarrow\infty$
for any $t\in\left[0,T\right]$, which establishes the second part of
Theorem 3.2.

\textbf{Step 3:} As we already proved in Step 2, $\tilde{B}_{t}^{\left(m\right)}\rightarrow B_{t}$
as $m\rightarrow\infty$, and since $\tilde{\psi}^{\left(m\right)}$
is uniformly integrable, it is straightforward to have $\tilde{e}_{0}^{\left(m\right)}\rightarrow e_{0}$
as $m\rightarrow\infty$.

\textbf{Step 4:} It is elemental to prove that $\tilde{P}_{0}^{\left(m\right)}\rightarrow P_{0}$
as $m\rightarrow\infty$ based on our results in Steps 1 and 3.

\section{Figures}

\begin{sidewaysfigure}[h!]
\centering
\caption{Approximated Exercise Boundary for Different Strikes and Volatilities}
\includegraphics[scale=0.4]{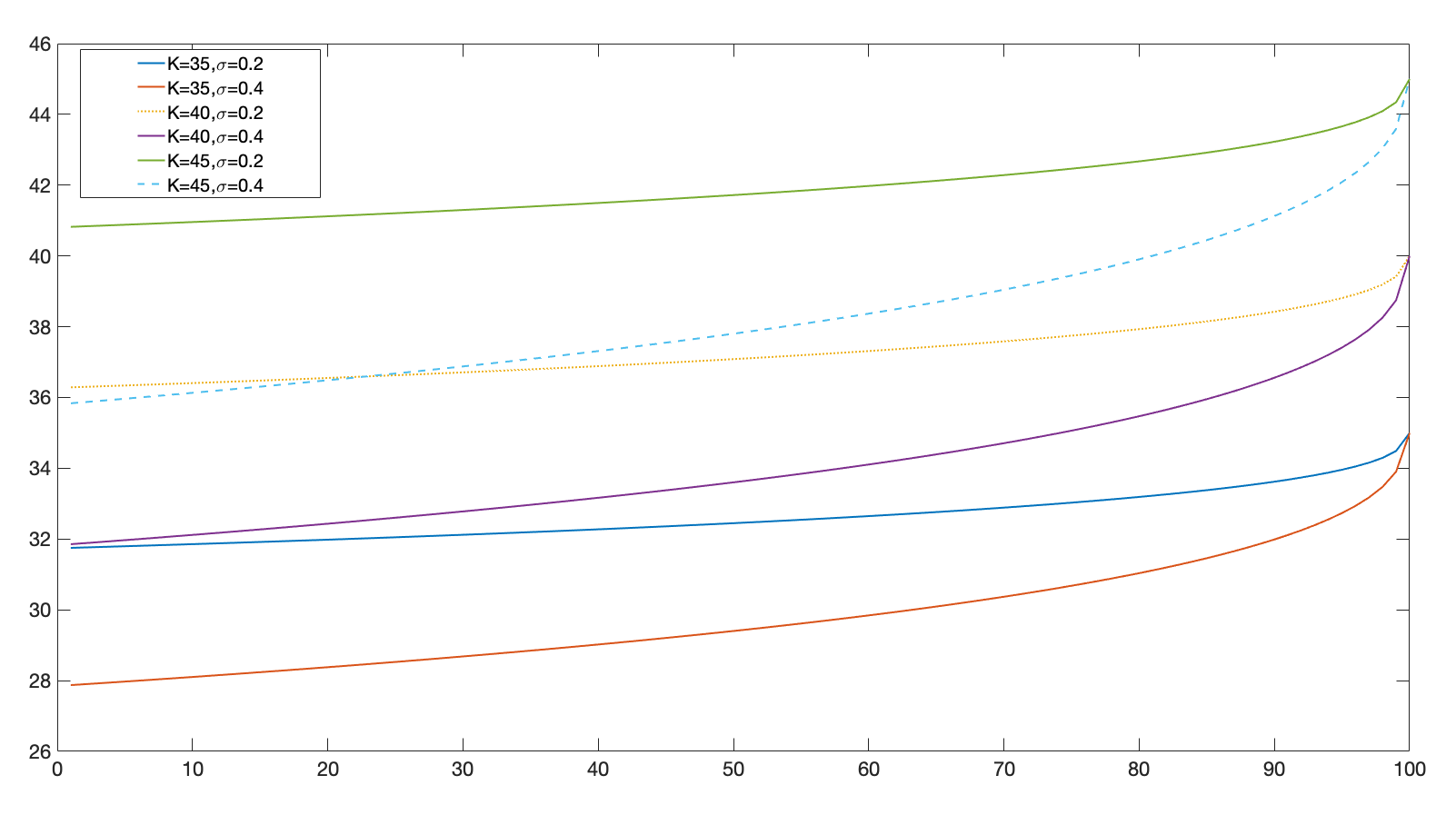}
\captionsetup{justification=raggedright,margin=0.45cm}
\caption*{\footnotesize{Note: The horizontal axis represents the 100 steps. That is, 100 in the horizontal axis means the time at maturity. The vertical axis represents the price.}}
\label{boundarygbm1}
\end{sidewaysfigure}

\begin{sidewaysfigure}[h!]
\centering
\caption{The Exercise Boundary: True vs. Approximation}
\includegraphics[scale=0.4]{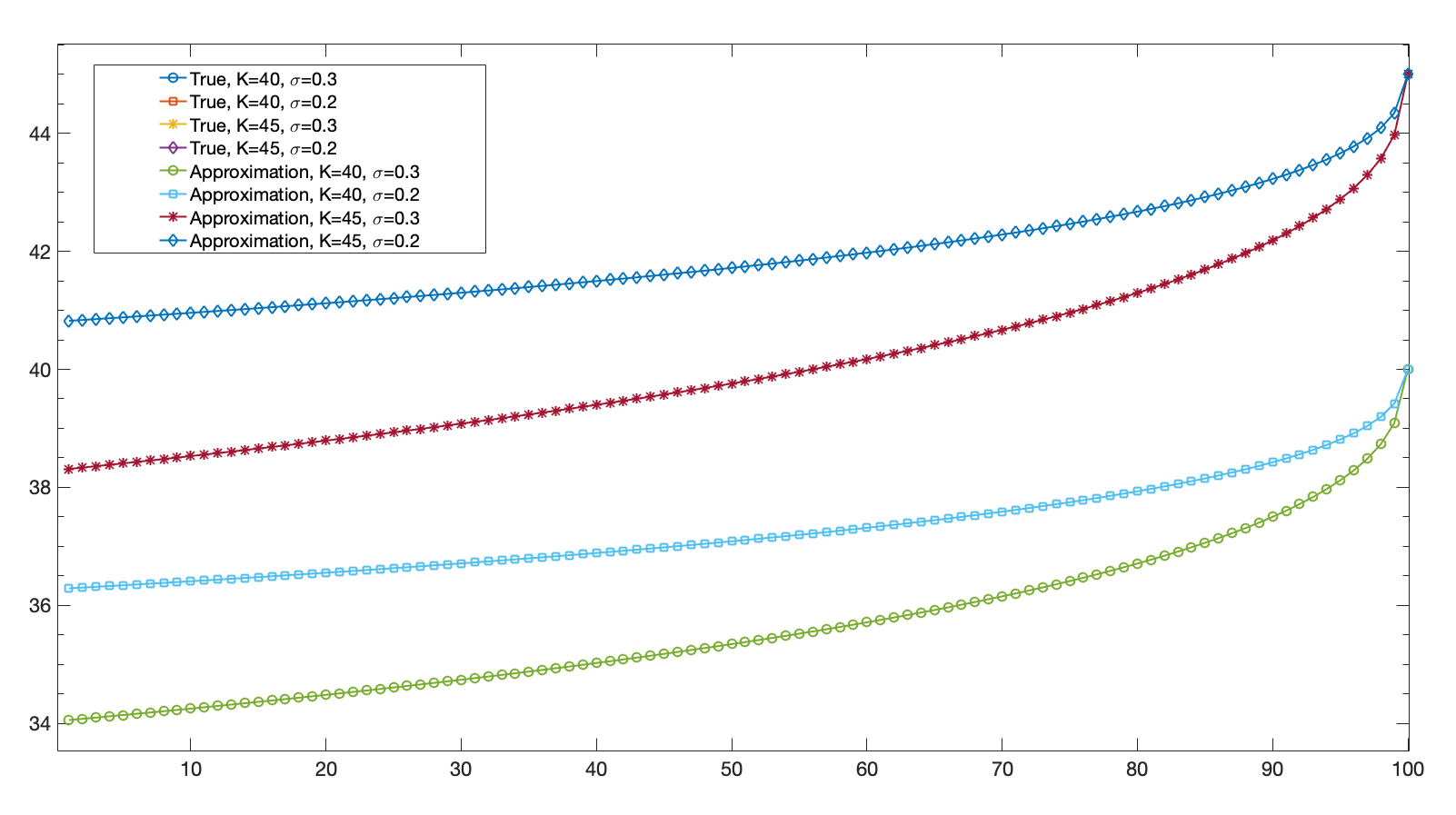}
\captionsetup{justification=raggedright,margin=0.45cm}
\caption*{\footnotesize{Note: The horizontal axis represents the 100 steps. That is, 100 in the horizontal axis means the time at maturity. The vertical axis represents the price. We use the same marker to represent the same set of parameter values. $\Diamond$ represents the boundary when $K=45$, and $\sigma=0.2$; $\star$ represents the boundary when $K=45$, and $\sigma=0.3$; $\Box$ represents the boundary when $K=40$, and $\sigma=0.2$;}; and $\circ$ represents the boundary when $K=40$, and $\sigma=0.3$;}
\label{boundarygbm2}
\end{sidewaysfigure}

\begin{sidewaysfigure}[h!]
\centering
\caption{The Exercise Boundary: Hermite polynomial Approximation vs. Finite Difference}
\includegraphics[scale=0.39]{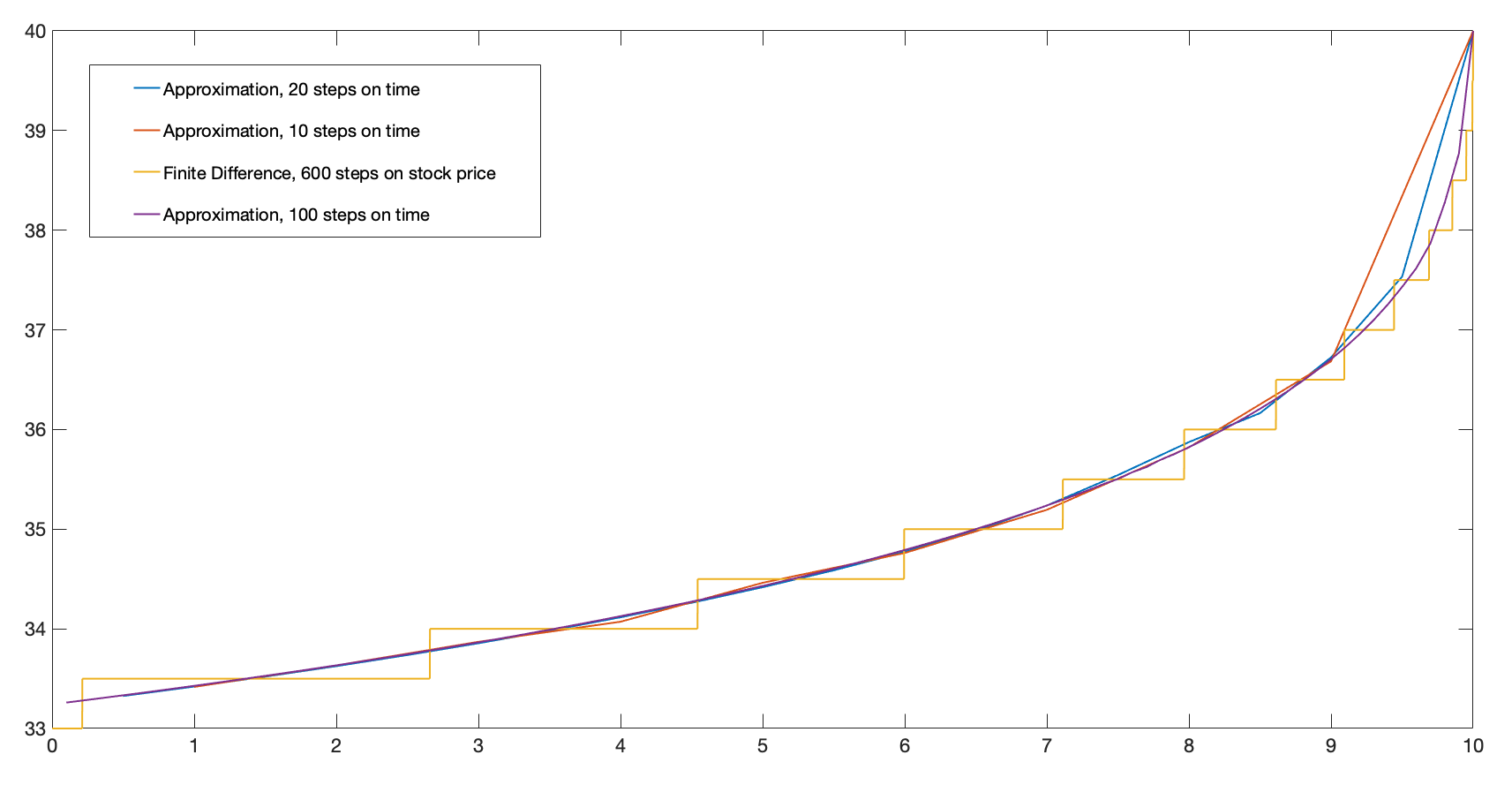}
\captionsetup{justification=raggedright,margin=0.45cm}
\caption*{\footnotesize{Note: The horizontal axis represents the 10 steps. That is, 10 in the horizontal axis means the time at maturity. The vertical axis represents the price. 1.932 seconds spent for our approach when we have 20 steps on time; 0.898 second spent when we have 10 steps on time; and 1.247 seconds spend for the finite difference approach when we have 600 steps on the support of stock price. We consider our approach with 100 steps on time as benchmark for comparison.}}
\label{fdvsh}
\end{sidewaysfigure}

\begin{figure}[h!]
\captionsetup{belowskip=12pt,aboveskip=4pt}
\caption{The Value of American Put and Strikes}
\begin{subfigure}{1\textwidth}
  \centering
  \captionsetup{justification=centering,margin=0.45cm}
    \caption{The value of American put with respect to different strikes for various orders of approximation}
  \includegraphics[scale=0.27]{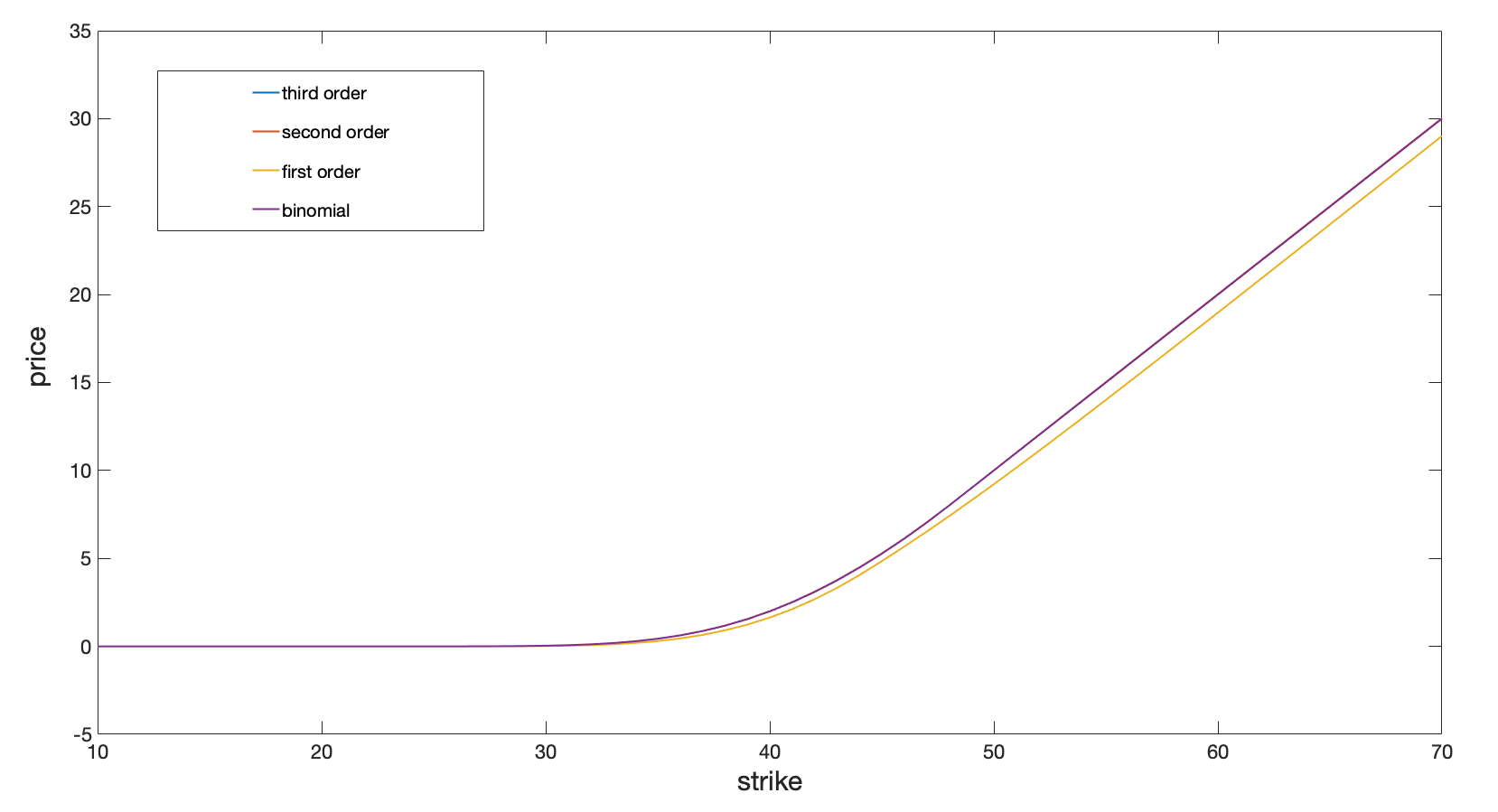}

  \label{strike}
\end{subfigure}
\begin{subfigure}{1\textwidth}
  \centering
  \captionsetup{justification=centering,margin=0.45cm}
    \caption{Approximation error of the American put value with respect to different strikes for various orders of approximation}
 \includegraphics[scale=0.27]{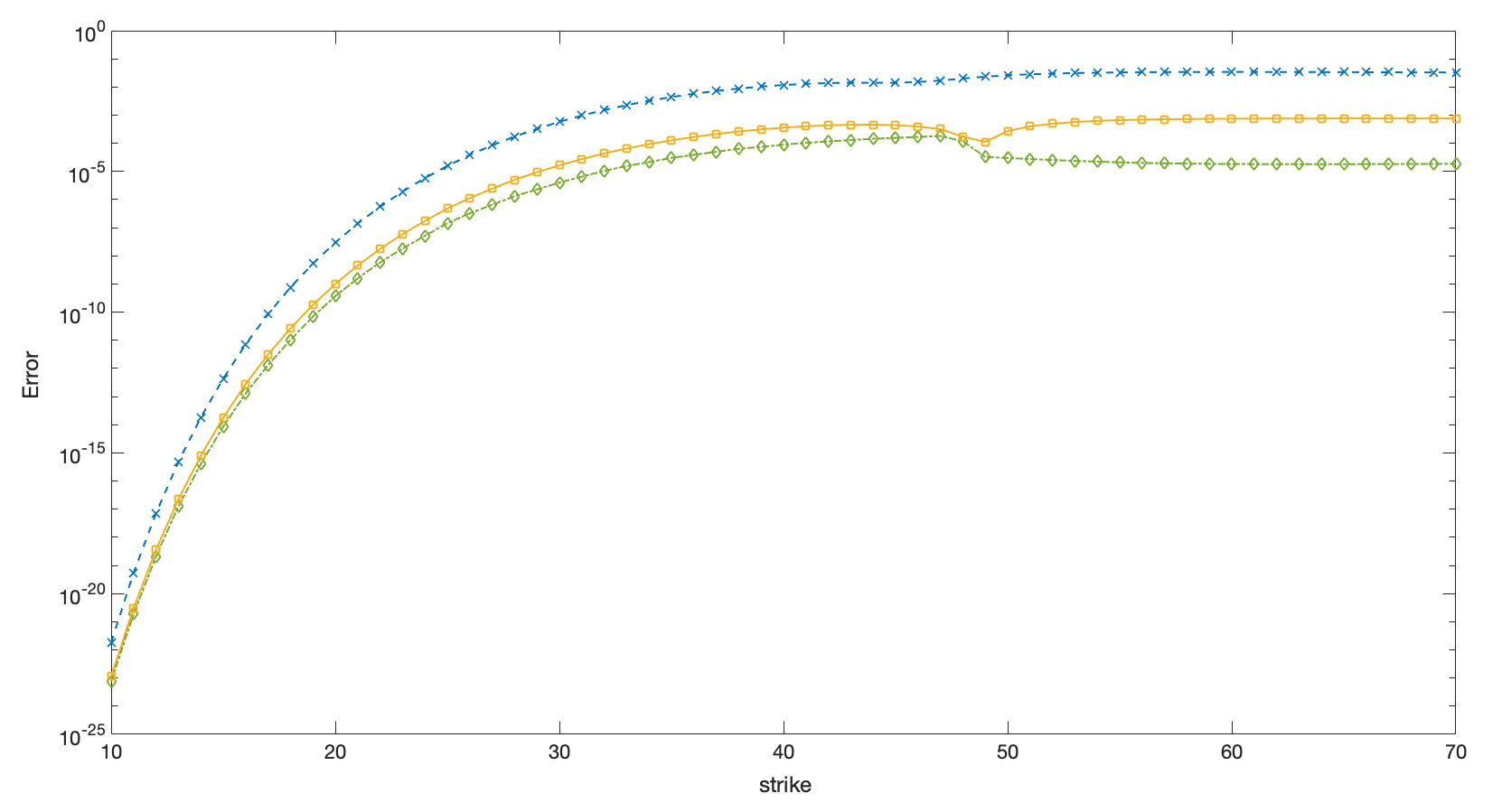}

  \label{strikeerror}
\end{subfigure}
\captionsetup{justification=raggedright,margin=0.45cm}
\caption*{\footnotesize{Note: (a) first order represents $m=1$ in the Hermite polynomial approximation of the transition density; second order represents $m=2$; third order represents $m=3$. (b) the blue curve is the relative error in approximating the price with first order accuracy of the approximation of the transition density; yellow curve is that for second order accuracy; green curve is that for third order accuracy}}
\label{strikeprice}
\end{figure}

\begin{sidewaysfigure}[h!]
\centering
\caption{The Exercise Boundary of American Put in the CEV Model}
\includegraphics[scale=0.4]{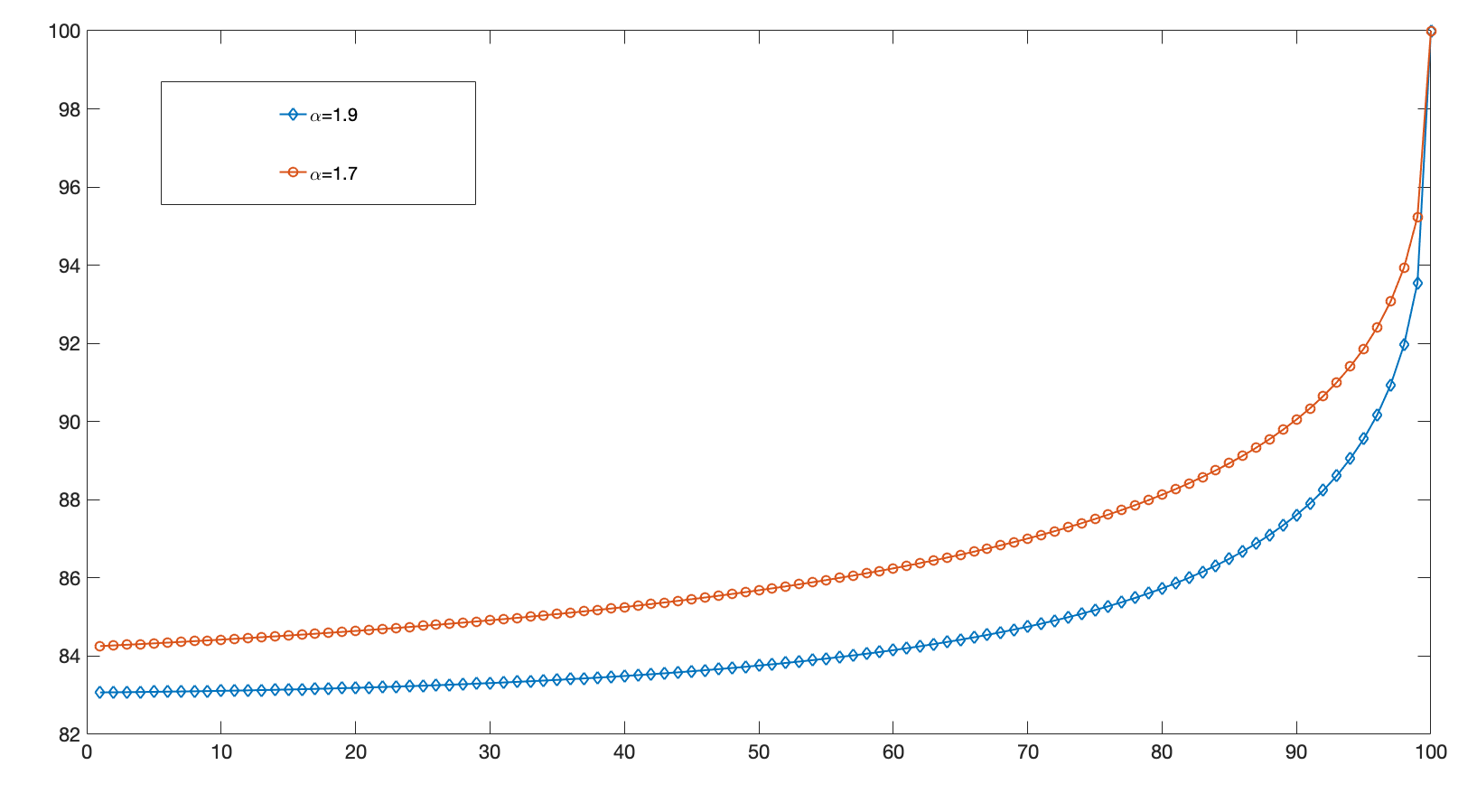}
\captionsetup{justification=raggedright,margin=0.45cm}
\caption*{\footnotesize{Note: The horizontal axis represents the 100 steps. That is, 100 in the horizontal axis means the time at maturity. The vertical axis represents the price. $K=100$, $r=6/100$, $\delta=r/2$, $\sigma=\sqrt{10}/5$, $S_{0}=40$, and $T=1$.}}
\label{boundarycev}
\end{sidewaysfigure}

\begin{sidewaysfigure}[h!]
\centering
\caption{The Exercise Boundary of American Put in the NMR Model}
\includegraphics[scale=0.4]{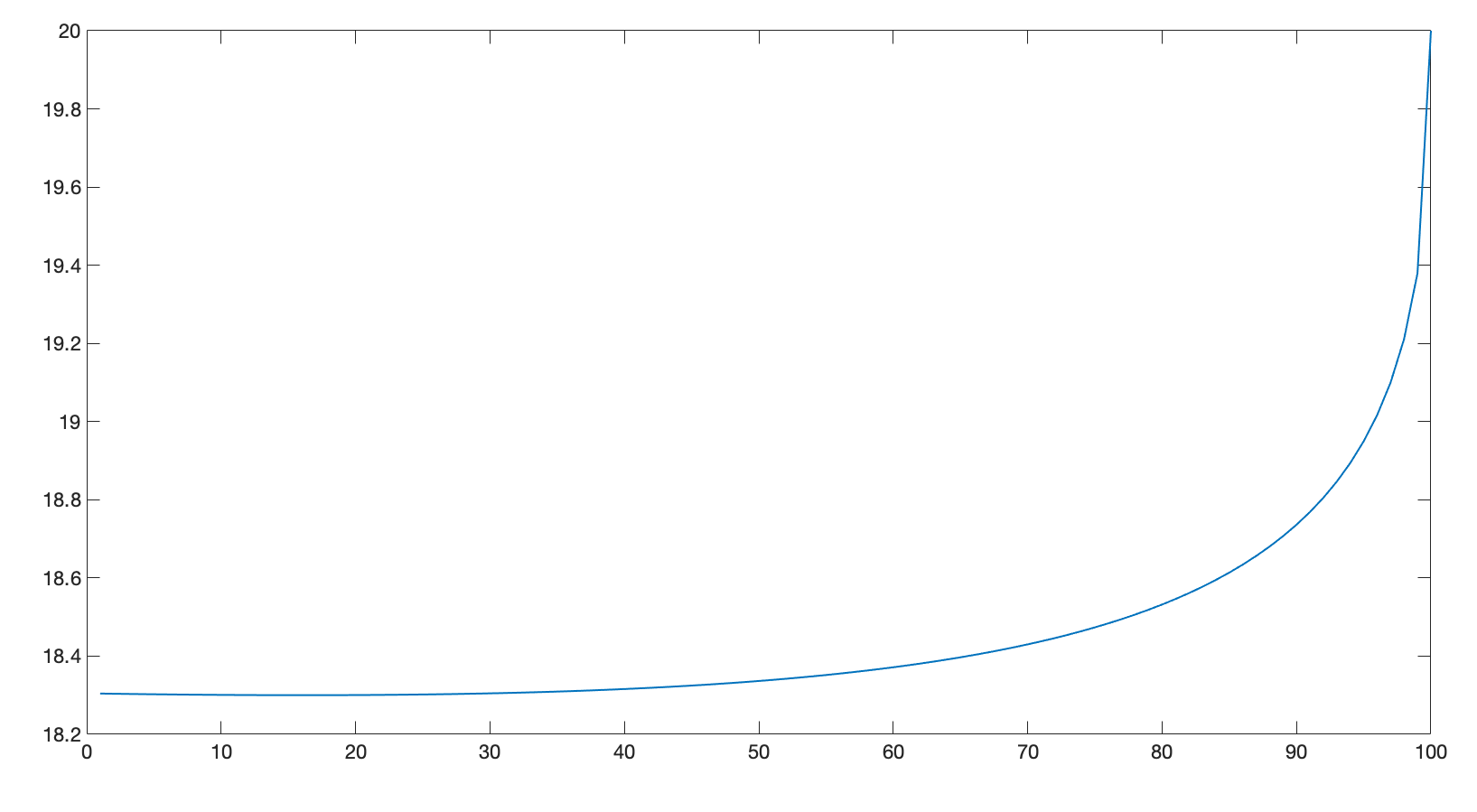}
\captionsetup{justification=raggedright,margin=0.45cm}
\caption*{\footnotesize{Note: The horizontal axis represents the 100 steps. That is, 100 in the horizontal axis means the time at maturity. The vertical axis represents the price. $a=500$, $b=5$, $c=0.05$, $v=-0.05$, $\sigma=0.2$, $\gamma=3/2$, $K=20$, $r=5/100$, $\delta=0$, $S_{0}=20$, and $T=0.0833$.}}
\label{nmr2}
\end{sidewaysfigure}

\begin{sidewaysfigure}[h!]
\centering
\caption{The Exercise Boundary of American Put in the DMR Model}
\includegraphics[scale=0.36]{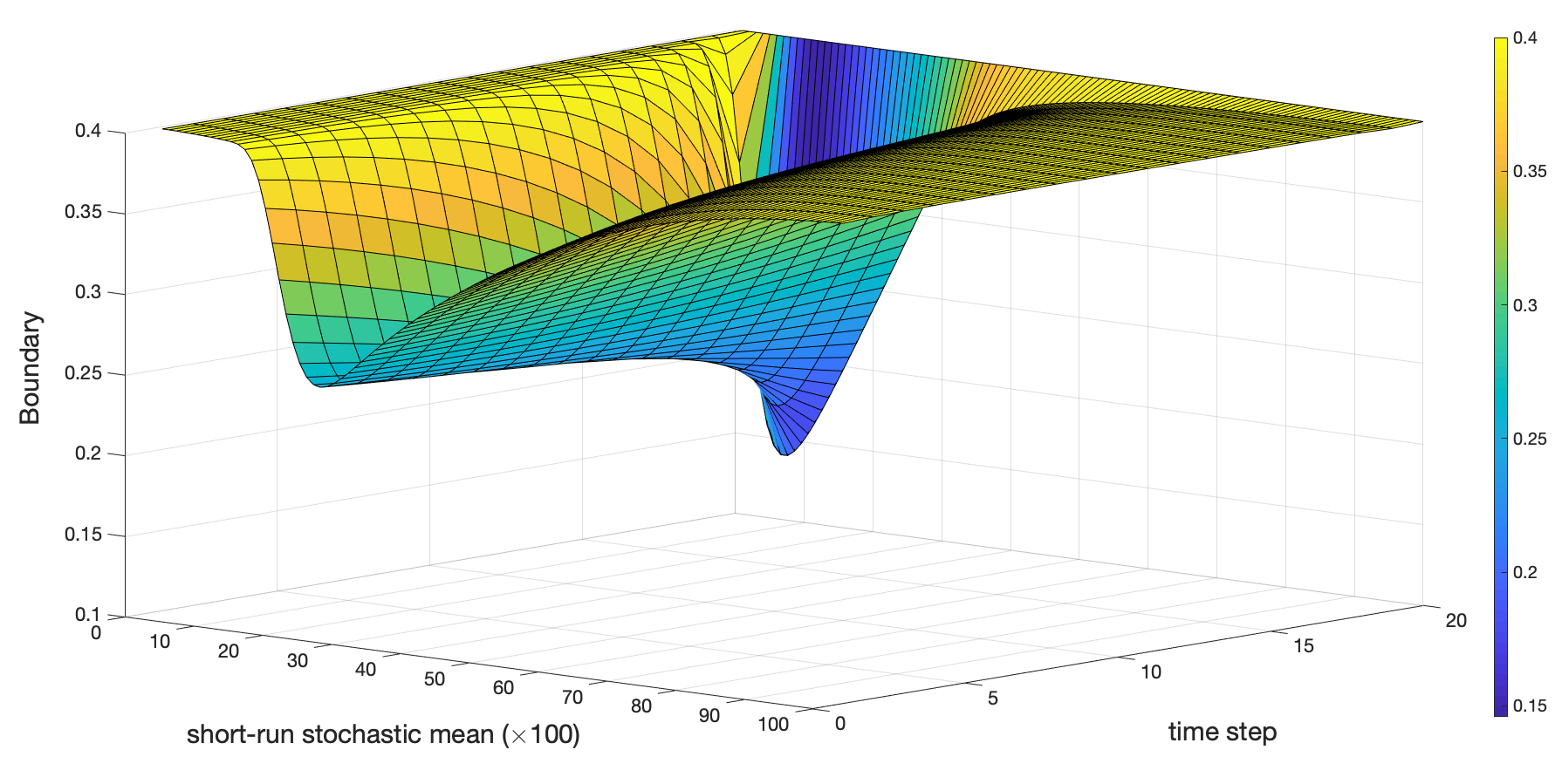}
\captionsetup{justification=raggedright,margin=0.45cm}
\caption*{\footnotesize{Note: The exercise boundary is approximated with 20 steps on time and 100 steps on $y$ for $y$ in $[0,1]$. $K=40/100$, $r=4.88/100$, $\delta=0$, $\sigma=0.25$, $\kappa=0.2$, $\beta=2.5$, $\xi=4$, $\alpha=0.25$, $T=0.5$}}
\label{dmr}
\end{sidewaysfigure}

\begin{sidewaysfigure}[h!]
\centering
\caption{The Exercise Boundary of American Put in Merton's Jump-Diffusion Model for $\lambda=1/100, 10/100,$ and $25/100$}
\includegraphics[scale=0.4]{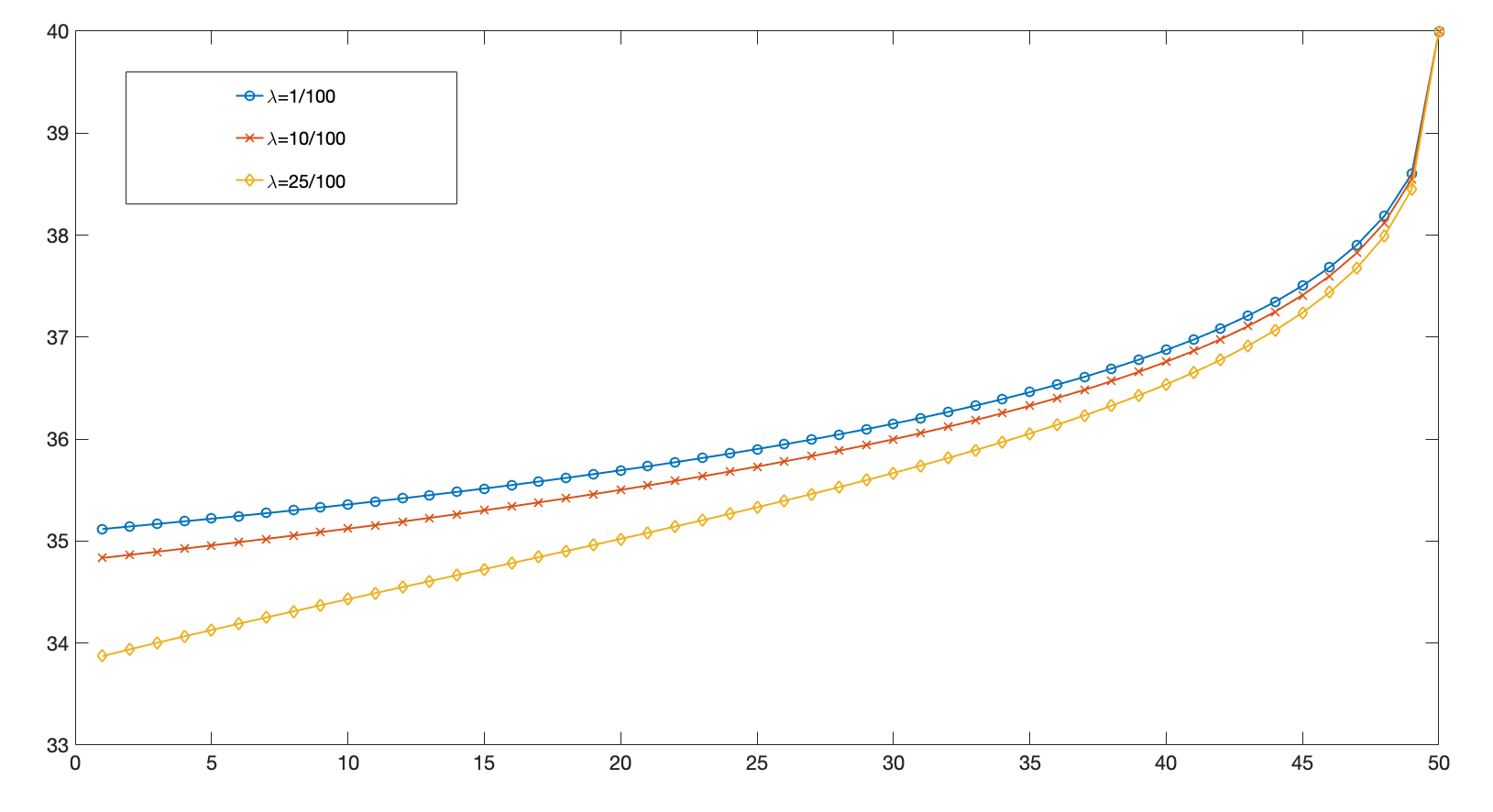}
\captionsetup{justification=raggedright,margin=0.45cm}
\caption*{\footnotesize{Note: The exercise boundary is approximated with 50 steps. $K=40$, $r=4.88/100$, $\sigma=0.2$, $\mu_{J}=0$, $\sigma_{J}=0.2$, $S_{0}=40$, and $T=0.5$.}}
\label{merton}
\end{sidewaysfigure}

\begin{sidewaysfigure}[h!]
\centering
\caption{The Exercise Boundary of American Put in Kou's Jump-Diffusion Model for $\lambda=1/100, 10/100,$ and $20/100$}
\includegraphics[scale=0.4]{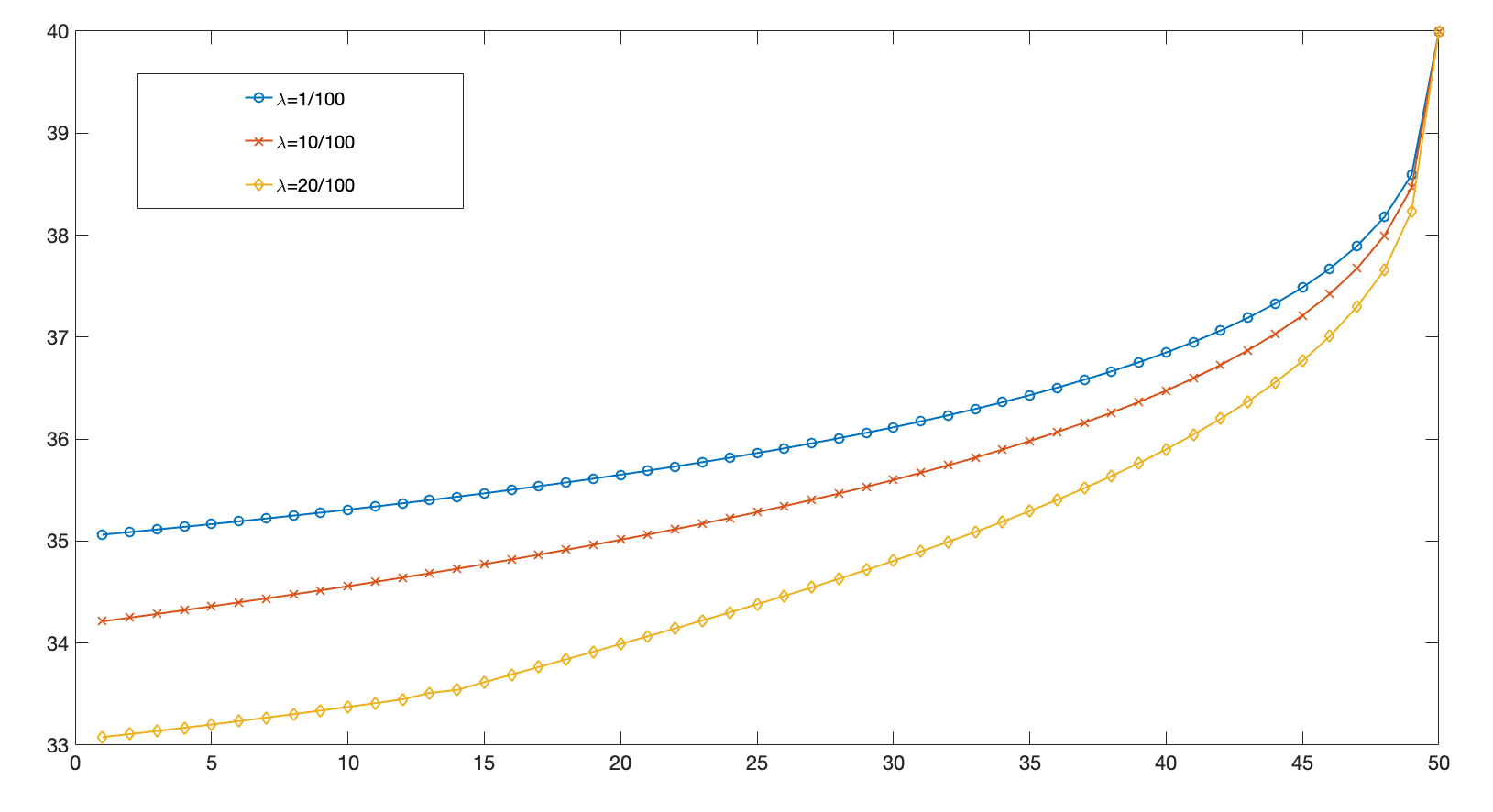}
\captionsetup{justification=raggedright,margin=0.45cm}
\caption*{\footnotesize{Note: The exercise boundary is approximated with 50 steps. $K=40$, $r=4.88/100$, $\delta=0$, $\sigma=0.2$, $p=0.04$, $q=0.96$, $\eta_{1}=3.7$, $\eta_{2}=1.8$, $S_{0}=40$, and $T=0.5$.}}
\label{kou}
\end{sidewaysfigure}

\end{document}